\documentclass[a4paper]{amsart}

\usepackage{amscd}
\usepackage{amsmath}
\usepackage{amssymb}
\usepackage{amsthm}
\usepackage{bbm}

\usepackage[T1]{fontenc}
\usepackage[utf8]{inputenc}

\newif\ifpdf
\ifx\pdfoutput\undefined
   \pdffalse        
\else
   \pdfoutput=1     
   \pdftrue
\fi

\ifpdf
   \usepackage[pdftex]{graphicx}
\else
   \usepackage{graphicx}
\fi

\numberwithin{equation}{section} \swapnumbers

\newtheorem{satz}{Satz}[section]

\newtheorem{theorem}[satz]{Theorem}
\newtheorem{proposition}[satz]{Proposition}

\newtheorem{lemma}[satz]{Lemma}

\newtheorem{definition}[satz]{Definition}

\newtheorem{remark}[satz]{Remark}

\begin{document}
\hyphenation{pre-sent cor-res-pond know-led-ge}

\title[Exponential stock models driven by tempered stable processes]{Exponential stock models driven by tempered stable processes}
\author{Uwe K{\"u}chler \and Stefan Tappe}
\address{Humboldt Universit\"{a}t zu Berlin, Institut f\"{u}r Mathematik, Unter den Linden 6, D-10099 Berlin, Germany}
\email{kuechler@mathematik.hu-berlin.de}
\address{Leibniz Universit\"{a}t Hannover, Institut f\"{u}r Mathematische Stochastik, Welfengarten 1, D-30167 Hannover, Germany}
\email{tappe@stochastik.uni-hannover.de}
\begin{abstract}
We investigate exponential stock models driven by tempered stable processes, which constitute a rich family of purely discontinuous L\'{e}vy processes. With a view of option pricing, we provide a systematic analysis of the existence of equivalent martingale measures, under which the model remains analytically tractable. This includes the existence of Esscher martingale measures and martingale measures having minimal distance to the physical probability measure. Moreover, we provide pricing formulae for European call options and perform a case study.
\end{abstract}
\maketitle

\section{Introduction}

Tempered stable distributions form a class of distributions that have attracted the interest of researchers from probability theory as well as financial mathematics. They have first been introduced in \cite{Koponen}, where the associated L\'{e}vy processes are called ``truncated L\'{e}vy flights'', and have been generalized by several authors. Tempered stable distributions form a six parameter family of infinitely divisible distributions, which cover several well-known subclasses like Variance Gamma distributions \cite{Madan-1990, Madan}, bilateral Gamma distributions \cite{Kuechler-Tappe, Kuechler-Tappe-shapes} and CGMY distributions \cite{CGMY}. Properties of tempered stable distributions have been investigated, e.g., in \cite{Rosinski, Zhang,  Sztonyk, Bianchi-3}, and in \cite{Kuechler-Tappe-TS}, where some of the results of this paper have been announced. For financial modeling they have been applied, e.g., in \cite{Boy, Cont-Tankov, Mercuri, Bianchi-1, Bianchi-2}, see also the recent textbook \cite{Bianchi-book}.

The purpose of this paper is to provide a systematic analysis of the existence of equivalent martingale measures for exponential stock price models driven by tempered stable processes, under which the computation of option prices remains analytically tractable. In particular, we are interested in martingale measures, under which the driving process remains a tempered stable process, or at least becomes a L\'{e}vy process for which the characteristic function is explicitly known.

Equivalent martingale measures of interest, under which the driving process remains a tempered stable process, are the Esscher martingale measure and bilateral Esscher martingale measures which minimize the distance to the original probability measure in a certain sense, for example the minimal entropy martingale measure or the $p$-optimal martingale measure. We will examine the existence of these martingale measures in detail. Furthermore, we will treat the F\"{o}llmer Schweizer minimal martingale measure. In case of existence, the driving process is the sum of two independent tempered stable processes under this measure, and thus the model remains analytically tractable. For all the just mentioned martingale measures, we will derive option pricing formulae. Moreover, we will illustrate our findings by means of a case study.

The remainder of this text is organized as follows: In Section~\ref{sec-stock} we introduce the stock model. Afterwards, in Section~\ref{sec-Esscher} we study Esscher transforms, in Section~\ref{sec-bilateral-Esscher} we study bilateral Esscher transforms, and in Section~\ref{sec-MMM} we treat the F\"{o}llmer Schweizer minimal martingale measure. Section~\ref{sec-option} is devoted to option pricing formulae, and in Section~\ref{sec-case-study} we provide the case study.

\section{Stock price models driven by tempered stable
processes}\label{sec-stock}

In this section, we shall introduce the stock price model and review some results about tempered stable processes. The reader is referred to \cite{Kuechler-Tappe-TS} for all results about tempered stable processes which we recall in this section.

Let $(\Omega,\mathcal{F},(\mathcal{F}_t)_{t \geq 0},\mathbb{P})$ be a
filtered probability space satisfying the usual conditions.
We fix parameters $\alpha^+,\lambda^+,\alpha^-,\lambda^- \in (0,\infty)$ and $\beta^+,\beta^- \in (0,1)$. An infinitely divisible distribution $\eta$ on $(\mathbb{R},\mathcal{B}(\mathbb{R}))$ is called a \emph{tempered stable distribution}, denoted
\begin{align*}
\eta = {\rm TS}(\alpha^+,\beta^+,\lambda^+;\alpha^-,\beta^-,\lambda^-),
\end{align*}
if its characteristic function is given by
\begin{align*}
\varphi(z) = \exp \bigg( \int_{\mathbb{R}} \big( e^{izx} - 1 \big) F(dx) \bigg), \quad z \in \mathbb{R} 
\end{align*}
where the L\'{e}vy measure $F$ is
\begin{align}\label{Levy-measure-TS}
F(dx) = \left( \frac{\alpha^+}{x^{1 + \beta^+}} e^{-\lambda^+ x}
\mathbbm{1}_{(0,\infty)}(x) + \frac{\alpha^-}{|x|^{1 + \beta^-}}
e^{-\lambda^- |x|} \mathbbm{1}_{(-\infty,0)}(x) \right)dx.
\end{align}

\begin{remark}
In \cite{Kuechler-Tappe-pricing} we have studied exponential stock models driven by bilateral Gamma processes, which would occur for $\beta^+ = \beta^- = 0$.
\end{remark}

We can express the characteristic function of $\eta$ as
\begin{equation}\label{cf-tempered-stable}
\begin{aligned}
\varphi(z) &= \exp \Big( \alpha^+ \Gamma(-\beta^+) \big[ (\lambda^+ - iz)^{\beta^+} - (\lambda^+)^{\beta^+} \big] 
\\ &\quad\quad\quad + \alpha^- \Gamma(-\beta^-) \big[ (\lambda^- + iz)^{\beta^-} - (\lambda^-)^{\beta^-} \big] \Big), \quad z \in \mathbb{R},
\end{aligned}
\end{equation}
where the powers stem from the main branch of the complex logarithm.
We call the L\'{e}vy process $X$ associated to $\eta$ a \emph{tempered stable process}, and write
\begin{align}\label{X-TS}
X \sim {\rm TS}(\alpha^+,\beta^+,\lambda^+;\alpha^-,\beta^-,\lambda^-).
\end{align}
The cumulant generating function 
\begin{align*}
\Psi(z) = \ln \mathbb{E}_{\mathbb{P}} [e^{zX_1}]
\end{align*}
exists on $[-\lambda^-,\lambda^+]$ and is given by
\begin{equation}\label{cumulant}
\begin{aligned}
\Psi(z) &= \alpha^+ \Gamma(-\beta^+) \big[ (\lambda^+ - z)^{\beta^+} -
(\lambda^+)^{\beta^+} \big]
\\ &\quad + \alpha^- \Gamma(-\beta^-)
\big[ (\lambda^- + z)^{\beta^-} - (\lambda^-)^{\beta^-} \big], \quad z \in
[-\lambda^-,\lambda^+].
\end{aligned}
\end{equation}
All increments of $X$ have a tempered stable distribution, more precisely
\begin{align}\label{distribution-TS-time}
X_t - X_s \sim {\rm TS}(\alpha^+ (t-s), \beta^+, \lambda^+; \alpha^- (t-s), \beta^-,
\lambda^-) \quad \text{for $0 \leq s < t$.}
\end{align}
A {\em tempered stable stock model} is an exponential L\'{e}vy model of the type
\begin{align}\label{exp-Levy-model}
\left\{
\begin{array}{rcl}
S_t & = & S_0 e^{X_t}
\\ B_t & = & e^{r t}
\end{array}
\right.
\end{align}
where $X$ denotes a tempered stable process and $S$ is a dividend paying stock with deterministic initial value $S_0 > 0$ and dividend rate $q \geq 0$. Furthermore, $B$ is the bank account with interest rate $r \geq 0$. In what follows, we assume that $r \geq q \geq 0$.
An equivalent probability measure $\mathbb{Q} \sim \mathbb{P}$ is a \emph{local martingale measure} (in short, \emph{martingale measure}), if the discounted stock price process
\begin{align}\label{discounted-price}
\tilde{S}_t := e^{-(r-q)t} S_t = S_0 e^{X_t - (r-q)t}, \quad t \geq
0
\end{align}
is a local $\mathbb{Q}$-martingale. 
The existence of a martingale measure $\mathbb{Q} \sim \mathbb{P}$ ensures that the stock market is free of arbitrage, and the price of an European option $\Phi(S_T)$, where $T > 0$ is the time of maturity and $\Phi : \mathbb{R} \rightarrow \mathbb{R}$ the payoff profile, is given by
\begin{align*}
\pi = e^{-rT} \mathbb{E}_{\mathbb{Q}}[\Phi(S_T)].
\end{align*}

\begin{lemma}\label{lemma-TS-mm}
The following statements are true:
\begin{enumerate}
\item If $\lambda^+ \geq 1$, then $\mathbb{P}$ is a martingale measure
if and only if
\begin{align}\label{martingale-eqn}
\alpha^+ \Gamma(-\beta^+) \big[ (\lambda^+ - 1)^{\beta^+} -
(\lambda^+)^{\beta^+} \big] + \alpha^- \Gamma(-\beta^-) \big[ (\lambda^- +
1)^{\beta^-} - (\lambda^-)^{\beta^-} \big] = r-q.
\end{align}

\item If $\lambda^+ < 1$, then $\mathbb{P}$ is never a martingale measure.
\end{enumerate}
\end{lemma}

\begin{proof}
By \cite[Lemma~2.6]{Kuechler-Tappe-pricing} the measure $\mathbb{P}$ is a martingale measure if and only if $\mathbb{E}_{\mathbb{P}}[e^{X_1}] = 1$, and hence, the assertion follows by taking into account (\ref{cumulant}).
\end{proof}

\section{Existence of Esscher martingale measures}\label{sec-Esscher}

In this section, we study the Esscher transform, which was
pioneered in \cite{Gerber}. Throughout this section, let $X$ be a tempered stable process of the form (\ref{X-TS}).

\begin{definition}
Let $\Theta \in (-\lambda^-,\lambda^+)$ be
arbitrary. The \emph{Esscher transform} $\mathbb{P}^{\Theta}$ is
defined as the locally equivalent probability measure with
likelihood process
\begin{align}\label{density-process}
\Lambda_t(\mathbb{P}^{\Theta},\mathbb{P}) := \frac{d
\mathbb{P}^{\Theta}}{d \mathbb{P}}\bigg|_{\mathcal{F}_t} = e^{\Theta
X_t - \Psi(\Theta) t}, \quad t \geq 0
\end{align}
where $\Psi$ denotes the cumulant generating function given by
(\ref{cumulant}).
\end{definition}

\begin{lemma}\label{lemma-TS-change-Esscher}
For every $\Theta \in (-\lambda^-,\lambda^+)$ we have 
\begin{align*}
X \sim {\rm TS}(\alpha^+,\beta^+,\lambda^+ - \Theta;\alpha^-,\beta^-,\lambda^- + \Theta) 
\end{align*}
under $\mathbb{P}^{\Theta}$.
\end{lemma}

\begin{proof}
This follows from Proposition~2.1.3 and Example~2.1.4 in \cite{Kuechler-Soerensen}.
\end{proof}

We define the function $f : [-\lambda^-,\lambda^+ - 1] \rightarrow \mathbb{R}$ as
\begin{align*}
f(\Theta) := f^+(\Theta) + f^-(\Theta),
\end{align*}
where we have set
\begin{align*}
f^+(\Theta) &:= \alpha^+ \Gamma(-\beta^+) \big[ (\lambda^+ - \Theta -
1)^{\beta^+} - (\lambda^+ - \Theta)^{\beta^+} \big],
\\ f^-(\Theta) &:= \alpha^-
\Gamma(-\beta^-) \big[ (\lambda^- + \Theta + 1)^{\beta^-} - (\lambda^- +
\Theta)^{\beta^-} \big].
\end{align*}

\begin{theorem}\label{thm-Esscher}
The following statements are true:
\begin{enumerate}
\item There exists $\Theta \in (-\lambda^-,\lambda^+)$
such that $\mathbb{P}^{\Theta}$ is a martingale measure if and only if
\begin{align}\label{cond-for-Esscher-lambda}
&\lambda^+ + \lambda^- > 1
\\ \label{cond-domain-temp}
\text{and} \quad &r-q \in (f(-\lambda^-),f(\lambda^+ - 1)].
\end{align}
\item Condition (\ref{cond-domain-temp}) is equivalent to
\begin{align*}
&\alpha^+ \Gamma(-\beta^+) \big[ (\lambda^+ + \lambda^- - 1)^{\beta^+} -
(\lambda^+ + \lambda^-)^{\beta^+} \big] + \alpha^- \Gamma(-\beta^-)
\\ &< r-q \leq -\alpha^+ \Gamma(-\beta^+) + \alpha^- \Gamma(-\beta^-)
\big[ (\lambda^+ + \lambda^-)^{\beta^-} - (\lambda^+ + \lambda^- -
1)^{\beta^-} \big].
\end{align*}
\item If conditions (\ref{cond-for-Esscher-lambda}) and (\ref{cond-domain-temp}) are
satisfied, then $\Theta$ is unique, belongs to the interval $(-\lambda^-,
\lambda^+ - 1]$, and it is the unique solution of the equation
\begin{align}\label{Esscher-equation}
f(\Theta) = r-q.
\end{align}
\end{enumerate}
\end{theorem}

\begin{proof}
Let $\Theta \in (-\lambda^-,\lambda^+)$ be arbitrary. In view of Lemmas~\ref{lemma-TS-change-Esscher} and \ref{lemma-TS-mm}, the probability measure $\mathbb{P}^{\Theta}$ is a martingale measure if and only if $\lambda^+ - \Theta \geq 1$, i.e. $\Theta \in (-\lambda^-,\lambda^+ - 1]$, and (\ref{Esscher-equation}) is fulfilled. Note that $(-\lambda^-,\lambda^+ - 1] \neq \emptyset$ if and only if (\ref{cond-for-Esscher-lambda}) is satisfied.
For the functions $f^+$ and $f^-$ we obtain the derivatives
\begin{align*}
(f^+)'(\Theta) &= -\alpha^+ \beta^+ \Gamma(-\beta^+) \big[ (\lambda^+ - \Theta - 1)^{\beta^+ - 1} - (\lambda^+ - \Theta)^{\beta^+ - 1} \big],
\\ (f^-)'(\Theta)  &= -\alpha^- \beta^- \Gamma(-\beta^-) \big[ (\lambda^- + \Theta)^{\beta^- - 1} - (\lambda^- + \Theta + 1)^{\beta^- - 1} \big]
\end{align*}
for $\Theta \in (-\lambda^-,\lambda^+ - 1]$.
Noting that $\beta^+,\beta^- \in (0,1)$, we see that $(f^+)',(f^-)' > 0$ on the interval $(-\lambda^-,\lambda^+ - 1]$.
Hence, $f$ is strictly increasing on $(-\lambda^-,\lambda^+ - 1]$, which completes the proof.
\end{proof}

\begin{remark}
In contrast to the present situation, for bilateral Gamma stock models ($\beta^+ = \beta^- = 0$) condition (\ref{cond-for-Esscher-lambda}) alone is already sufficient for the existence of an Esscher martingale measure, cf. \cite[Remark~4.4]{Kuechler-Tappe-pricing}.
\end{remark}

\section{Existence of minimal distance measures preserving the class of
tempered stable processes}\label{sec-bilateral-Esscher}

In the literature, one often performs option pricing by finding an
equivalent martingale measure $\mathbb{Q} \sim \mathbb{P}$ which
minimizes the distance
\begin{align*}
\mathbb{E}_{\mathbb{P}}[ g (\Lambda_1(\mathbb{Q},\mathbb{P})) ]
\end{align*}
for some strictly convex function $g : (0,\infty) \rightarrow
\mathbb{R}$. Here are popular choices for the function $g$:
\begin{itemize}
\item For $g(x) = x \ln x$ we call $\mathbb{Q}$ the \emph{minimal entropy martingale measure}. 

\item For $g(x) = x^p$ with $p > 1$ we call $\mathbb{Q}$ the \emph{$p$-optimal martingale measure}.

\item For $p = 2$ we call $\mathbb{Q}$ the \emph{variance-optimal martingale measure}.
\end{itemize}
We refer to \cite[Section~5]{Kuechler-Tappe-pricing} for further remarks and related literature. While $p$-optimal equivalent martingale measures do not exist in tempered stable stock models (which follows from \cite[Example~2.7]{Bender}), we have the following result concerning the existence of minimal entropy martingale measures:

\begin{theorem}
The following statements are true:
\begin{enumerate}
\item If $\lambda^+ < 1$, then a minimal entropy martingale measure exists.

\item If $\lambda^+ \geq 1$, then a minimal entropy measure exists if and only if
\begin{align*}
\alpha^+ \Gamma(-\beta^+) \big[ (\lambda^+ - 1)^{\beta^+} -
(\lambda^+)^{\beta^+} \big] + \alpha^- \Gamma(-\beta^-)
\big[ (\lambda^- + 1)^{\beta^-} - (\lambda^-)^{\beta^-} \big] \geq r-q.
\end{align*}
\end{enumerate}
\end{theorem}

This result, which has been indicated in \cite[Remark~5.4]{Kuechler-Tappe-pricing}, follows by adjusting the arguments of the proof of \cite[Theorem~5.3]{Kuechler-Tappe-pricing} to the present situation, where the stock model is driven by a tempered stable process.

In this section, we shall minimize the relative entropy 
\begin{align*}
\mathbb{H}(\mathbb{Q} \,|\, \mathbb{P}) := \mathbb{E}_{\mathbb{P}}
[\Lambda_1(\mathbb{Q},\mathbb{P}) \ln
\Lambda_1(\mathbb{Q},\mathbb{P})] = \mathbb{E}_{\mathbb{Q}}[\ln
\Lambda_1(\mathbb{Q},\mathbb{P})]
\end{align*}
within the class of tempered stable processes by performing \textit{bilateral} Esscher transforms. Let $X$ be a tempered stable process of the form (\ref{X-TS}). We decompose the tempered stable process $X = X^+ - X^-$ as the difference of two independent subordinators. Their respective cumulant generating functions are given by
\begin{align}\label{cumulant-plus}
\Psi^+(z) &= \alpha^+ \Gamma(-\beta^+) \big[ (\lambda^+ - z)^{\beta^+} -
(\lambda^+)^{\beta^+} \big], \quad z \in (-\infty,\lambda^+],
\\ \label{cumulant-minus} \Psi^-(z) &= \alpha^- \Gamma(-\beta^-) \big[ (\lambda^- - z)^{\beta^-} -
(\lambda^-)^{\beta^-} \big], \quad z \in (-\infty,\lambda^-],
\end{align}
see~\cite{Kuechler-Tappe-TS}.
Note that $\Psi(z) = \Psi^+(z) + \Psi^-(-z)$ for $z \in [-\lambda^-,\lambda^+]$.

\begin{definition}\label{def-bilateral-Esscher}
Let $\theta^+ \in (-\infty,\lambda^+)$ and
$\theta^- \in (-\infty,\lambda^-)$ be arbitrary. The {\rm bilateral
Esscher transform} $\mathbb{P}^{(\theta^+,\theta^-)}$ is defined as
the locally equivalent probability measure with likelihood process
\begin{align*}
\Lambda_t(\mathbb{P}^{(\theta^+,\theta^-)},\mathbb{P}) := \frac{d
\mathbb{P}^{(\theta^+,\theta^-)}}{d
\mathbb{P}}\bigg|_{\mathcal{F}_t} = e^{\theta^+ X_t^+ -
\Psi^+(\theta^+) t} \cdot e^{\theta^- X_t^- - \Psi^-(\theta^-) t},
\quad t \geq 0.
\end{align*}
\end{definition}

Note that the Esscher transforms $\mathbb{P}^{\Theta}$ from Section~\ref{sec-Esscher} are special cases of the just introduced bilateral
Esscher transforms $\mathbb{P}^{(\theta^+,\theta^-)}$. Indeed, we have
\begin{align}\label{special-case}
\mathbb{P}^{\Theta} = \mathbb{P}^{(\Theta,-\Theta)}, \quad \Theta
\in (-\lambda^-,\lambda^+).
\end{align}

\begin{lemma}\label{lemma-two-sided-preserving}
For all $\theta^+ \in (-\infty,\lambda^+)$ and
$\theta^- \in (-\infty,\lambda^-)$ we have 
\begin{align*}
X \sim {\rm TS}(\alpha^+,\beta^+,\lambda^+ - \theta^+;\alpha^-,\beta^-,\lambda^- -
\theta^-)
\end{align*}
under $\mathbb{P}^{(\theta^+,\theta^-)}$.
\end{lemma}

\begin{proof}
This follows from Proposition~2.1.3 and Example~2.1.4 in \cite{Kuechler-Soerensen}.
\end{proof}

\begin{proposition}\label{prop-mm-bilateral}
The following statements are true:
\begin{enumerate}
\item If we have
\begin{align}\label{ineqn-rq-1}
-\alpha^+ \Gamma(-\beta^+) \leq r-q, 
\end{align}
then no pair $(\theta^+,\theta^-) \in (-\infty,\lambda^+) \times (-\infty,\lambda^-)$ with $\mathbb{P}^{(\theta^+,\theta^-)}$ being a martingale measure exists.

\item If we have
\begin{align}\label{ineqn-rq-2}
-\alpha^+ \Gamma(-\beta^+) > r-q, 
\end{align}
then there exist $-\infty \leq \theta_1^+ < \theta_2^+ \leq \lambda^+ - 1$ and a continuous, strictly increasing, bijective function $\Phi : (\theta_1^+,\theta_2^+) \rightarrow (-\infty,\lambda^-)$ such that:
\begin{itemize}
\item For all $\theta^+ \in (\theta_1^+,\theta_2^+)$ there exists a unique $\theta^- \in (-\infty,\lambda^-)$ with $\mathbb{P}^{(\theta^+,\theta^-)}$ being a martingale measure, and it is given by $\theta^- = \Phi(\theta^+)$.

\item For all $\theta^+ \in (-\infty,\lambda^+) \setminus (\theta_1^+,\theta_2^+)$ no $\theta^- \in (-\infty,\lambda^-)$ with $\mathbb{P}^{(\theta^+,\theta^-)}$ being a martingale measure exists.

\end{itemize}

\end{enumerate}
\end{proposition}

\begin{proof}
We introduce the functions $f^+ : (-\infty,\lambda^+ - 1] \rightarrow \mathbb{R}$ and $f^- : (-\infty,\lambda^-] \rightarrow \mathbb{R}$ as
\begin{align*}
f^+(\theta^+) &:= \alpha^+ \Gamma(-\beta^+) \big[ (\lambda^+ - \theta^+ - 1)^{\beta^+} - (\lambda^+ - \theta^+)^{\beta^+} \big],
\\ f^-(\theta^-) &:= \alpha^- \Gamma(-\beta^-) \big[ (\lambda^- - \theta^- + 1)^{\beta^-} - (\lambda^- - \theta^-)^{\beta^-} \big].
\end{align*}
By Lemmas~\ref{lemma-TS-mm} and \ref{lemma-two-sided-preserving}, the measure
$\mathbb{P}^{(\theta^+,\theta^-)}$ is a martingale measure if and only if $\theta^+ \in (-\infty,\lambda^+ - 1]$ and
\begin{align}\label{mart-eqn-two-sided}
f^+(\theta^+) + f^-(\theta^-) = r-q.
\end{align}
The function $f^+$ is continuous and strictly increasing on $(-\infty,\lambda^+ - 1]$ with 
\begin{align*}
\lim_{\theta^+ \rightarrow -\infty} f^+(\theta^+) = 0 \quad \text{and} \quad f^+(\lambda^+ - 1) = -\alpha^+ \Gamma(-\beta^+) > 0.
\end{align*}
The function $f^-$ is continuous and strictly decreasing on $(-\infty,\lambda^-]$ with 
\begin{align*}
\lim_{\theta^- \rightarrow -\infty} f^-(\theta^-) = 0 \quad \text{and} \quad f^-(\lambda^-) = \alpha^- \Gamma(-\beta^-) < 0.
\end{align*}
Therefore, if we have (\ref{ineqn-rq-1}), then for no pair $(\theta^+,\theta^-) \in (-\infty,\lambda^+ - 1] \times (-\infty,\lambda^-)$ equation (\ref{mart-eqn-two-sided}) is satisfied. If we have (\ref{ineqn-rq-2}), then let $-\infty \leq \theta_1^+ < \theta_2^+ \leq \lambda^+ - 1$ be the unique solutions of the equations
\begin{align*}
f^+(\theta_1^+) &= r-q,
\\ f^+(\theta_2^+) &= r-q - \alpha^- \Gamma(-\beta^-),
\end{align*}
with the conventions 
\begin{align*}
\theta_1^+ = -\infty, \quad &\text{if $r-q = 0$,} 
\\ \theta_2^+ = \lambda^+ - 1, \quad &\text{if $r-q - \alpha^- \Gamma(-\beta^-) > -\alpha^+ \Gamma(-\beta^+)$,} 
\end{align*}
and define
\begin{align}\label{def-Phi}
\Phi(\theta^+) := (f^-)^{-1} ( r-q - f^+(\theta^+) ), \quad \theta^+ \in (\theta_1^+,\theta_2^+).
\end{align}
Then $\Phi$ is continuous and strictly increasing with $\Phi((\theta_1^+,\theta_2^+)) = (-\infty,\lambda^-)$, which finishes the proof.
\end{proof}

\begin{remark}\label{rem-measures}
The proof of Proposition~\ref{prop-mm-bilateral} shows that the situation $\theta_1^+ = -\infty$ occurs if and only if $r = q$ and that the situation $\theta_2^+ = \lambda^+ - 1$ occurs if and only if
\begin{align*}
r-q \leq -\alpha^+ \Gamma(-\beta^+) + \alpha^- \Gamma(-\beta^-).
\end{align*}
All equivalent measure transformations preserving the class of tempered stable processes are bilateral Esscher transforms; this follows from \cite[Proposition~8.1]{Kuechler-Tappe-TS}, see also \cite[Example~9.1]{Cont-Tankov}. Hence, we introduce the set of parameters
\begin{align*}
\mathcal{M}_{\mathbb{P}} := \{ (\theta^+, \theta^-) \in
(-\infty,\lambda^+) \times (-\infty,\lambda^-) \,|\,
\text{$\mathbb{P}^{(\theta^+,\theta^-)}$ is a martingale measure} \}
\end{align*}
such that the bilateral Esscher transform is a martingale measure. The
previous Proposition~\ref{prop-mm-bilateral} tells us that for (\ref{ineqn-rq-1}) we have $\mathcal{M}_{\mathbb{P}} = \emptyset$, and that for (\ref{ineqn-rq-2}) we have
\begin{align}\label{set-MM}
\mathcal{M}_{\mathbb{P}} = \{ (\theta,\Phi(\theta)) \in \mathbb{R}^2
\,|\, \theta \in (\theta_1^+,\theta_2^+) \}.
\end{align}
Moreover, we remark that condition (\ref{ineqn-rq-2}) is always fulfilled for $r=q$.
\end{remark}

\begin{lemma}\label{lemma-mm-TS-entropy}
For all $(\theta^+,\theta^-) \in (-\infty,\lambda^+) \times (-\infty,\lambda^-)$ we have
\begin{align*}
&\mathbb{H}(\mathbb{P}^{(\theta^+,\theta^-)} \,|\, \mathbb{P})
\\ &= -\alpha^+ \Gamma(-\beta^+) \Big( \lambda^+ \beta^+ (\lambda^+ - \theta^+)^{\beta^+ - 1} + (1-\beta^+) (\lambda^+ - \theta^+)^{\beta^+} - (\lambda^+)^{\beta^+} \Big)
\\ &\quad -\alpha^- \Gamma(-\beta^-) \Big( \lambda^- \beta^- (\lambda^- - \theta^-)^{\beta^- - 1} + (1-\beta^-) (\lambda^- - \theta^-)^{\beta^-} - (\lambda^-)^{\beta^-} \Big).
\end{align*}
\end{lemma}

\begin{proof}
The relative entropy of the bilateral Esscher transform is given by
\begin{align*}
\mathbb{H}(\mathbb{P}^{(\theta^+,\theta^-)} \,|\, \mathbb{P}) &= \mathbb{E}_{\mathbb{P}^{(\theta^+,\theta^-)}}[\ln \Lambda_1(\mathbb{P}^{(\theta^+,\theta^-)},\mathbb{P})] 
\\ &= \mathbb{E}_{\mathbb{P}^{(\theta^+,\theta^-)}}[\theta^+ X_1^+ - \Psi^+(\theta^+)] + \mathbb{E}_{\mathbb{P}^{(\theta^+,\theta^-)}}[\theta^- X_1^- - \Psi^-(\theta^-)].
\end{align*}
Using Lemma~\ref{lemma-two-sided-preserving} and \cite[Remark~2.7]{Kuechler-Tappe-TS} we obtain
\begin{align*}
&\mathbb{E}_{\mathbb{P}^{(\theta^+,\theta^-)}}[\theta^+ X_1^+ - \Psi^+(\theta^+)]
\\ &= \theta^+ \Gamma(1-\beta^+) \frac{\alpha^+}{(\lambda^+ - \theta^+)^{1-\beta^+}} - \alpha^+ \Gamma(-\beta^+) \left[ (\lambda^+ - \theta^+)^{\beta^+} - (\lambda^+)^{\beta^+} \right]
\\ &= -\alpha^+ \Gamma(-\beta^+) \bigg( \frac{\theta^+ \beta^+}{(\lambda^+ - \theta^+)^{1-\beta^+}} + (\lambda^+ - \theta^+)^{\beta^+} - (\lambda^+)^{\beta^+} \bigg)
\\ &= -\alpha^+ \Gamma(-\beta^+) \Big( (\beta^+ \theta^+ + \lambda^+ - \theta^+) (\lambda^+ - \theta^+)^{\beta^+ - 1} - (\lambda^+)^{\beta^+} \Big)
\\ &= -\alpha^+ \Gamma(-\beta^+) \Big( \big( \lambda^+ \beta^+ + (1-\beta^+) (\lambda^+ - \theta^+) \big) (\lambda^+ - \theta^+)^{\beta^+ - 1} - (\lambda^+)^{\beta^+} \Big)
\\ &= -\alpha^+ \Gamma(-\beta^+) \Big( \lambda^+ \beta^+ (\lambda^+ - \theta^+)^{\beta^+ - 1} + (1-\beta^+) (\lambda^+ - \theta^+)^{\beta^+} - (\lambda^+)^{\beta^+} \Big).
\end{align*}
An analogous calculation for $\mathbb{E}_{\mathbb{P}^{(\theta^+,\theta^-)}}[\theta^- X_1^- - \Psi^-(\theta^-)]$ finishes the proof.
\end{proof}

\begin{theorem}\label{thm-min-entropy-bilateral}
The following statements are true:
\begin{enumerate}
\item If (\ref{ineqn-rq-1}) is satisfied, then we have $\mathcal{M}_{\mathbb{P}} = \emptyset$.

\item If (\ref{ineqn-rq-2}) is satisfied, then there exist
$\theta^+ \in (-\infty,\lambda^+)$ and $\theta^- \in
(-\infty,\lambda^-)$ such that
\begin{align}\label{min-entropy-bilateral}
\mathbb{H}(\mathbb{P}^{(\theta^+,\theta^-)} \,|\, \mathbb{P}) =
\min_{(\vartheta^+, \vartheta^-) \in \mathcal{M}_{\mathbb{P}}}
\mathbb{H}(\mathbb{P}^{(\vartheta^+,\vartheta^-)} \,|\, \mathbb{P}).
\end{align}
\end{enumerate}
\end{theorem}

\begin{proof}
If (\ref{ineqn-rq-1}) is satisfied, then by Proposition~\ref{prop-mm-bilateral} we have $\mathcal{M}_{\mathbb{P}} = \emptyset$. Now, suppose that (\ref{ineqn-rq-2}) is satisfied, and let $\Phi : (\theta_1^+,\theta_2^+) \rightarrow (-\infty,\lambda^-)$ be the function from Proposition~\ref{prop-mm-bilateral}. Let $f : (\theta_1^+,\theta_2^+) \rightarrow \mathbb{R}$ be the function
\begin{align*}
&f(\theta) := -\alpha^+ \Gamma(-\beta^+) \Big( \lambda^+ \beta^+ (\lambda^+ - \theta)^{\beta^+ - 1} + (1-\beta^+) (\lambda^+ - \theta)^{\beta^+} - (\lambda^+)^{\beta^+} \Big)
\\ &-\alpha^- \Gamma(-\beta^-) \Big( \lambda^- \beta^- (\lambda^- - \Phi(\theta))^{\beta^- - 1} + (1-\beta^-) (\lambda^- - \Phi(\theta))^{\beta^-} - (\lambda^-)^{\beta^-} \Big).
\end{align*}
By Proposition~\ref{prop-mm-bilateral} and Lemma~\ref{lemma-mm-TS-entropy}, for each $\theta \in (\theta_1^+,\theta_2^+)$ the measure $\mathbb{P}^{(\theta,\Phi(\theta))}$ is a martingale measure and we have $\mathbb{H}(\mathbb{P}^{(\theta,\Phi(\theta))} \,|\, \mathbb{P}) = f(\theta)$.
The function $\Phi$ is strictly increasing with
\begin{align*}
\lim_{\theta \downarrow \theta_1^+} \Phi(\theta) = -\infty \quad \text{and} \quad \lim_{\theta \uparrow \theta_2^+} \Phi(\theta) = \lambda^-,
\end{align*}
which gives us
\begin{align*}
\lim_{\theta \downarrow \theta_1^+} f(\theta) = \infty \quad \text{and} \quad \lim_{\theta \uparrow \theta_2^+} f(\theta) = \infty.
\end{align*}
Since $f$ is continuous, it attains a minimum and the assertion follows.
\end{proof}

\begin{remark}
In contrast to bilateral Gamma stock models, it can happen that $\mathcal{M}_{\mathbb{P}} = \emptyset$, i.e., there is no equivalent martingale measure under which $X$ remains a tempered stable process. Moreover, in contrast to bilateral Gamma stock models, the function $\Phi$ from Proposition~\ref{prop-mm-bilateral}, which is defined in (\ref{def-Phi}) by means of the inverse of $f^-$, does not seem to be available in closed form, cf. \cite[Remark~6.7]{Kuechler-Tappe-pricing}.
\end{remark}

Next, we consider the $p$-distances
\begin{align*}
\mathbb{H}_p(\mathbb{Q} \,|\, \mathbb{P}) := \mathbb{E}_{\mathbb{P}} \Bigg[ \bigg( \frac{d \mathbb{Q}}{d \mathbb{P}} \bigg)^p \Bigg] \quad \text{for $p > 1$.}
\end{align*}
As mentioned at the beginning of this section, for tempered stable stock models the $p$-optimal martingale
measure does not exist. However, we can, as provided for the minimal entropy martingale measure, determine
the $p$-optimal martingale measure within the class of tempered stable processes. For this purpose, we compute the $p$-distance of a bilateral Esscher transform. Since the subordinators $X^+$ and $X^-$ are independent, for $p > 1$ and $\theta^+ \in (-\infty,\frac{\lambda^+}{p})$, $\theta^- \in (-\infty,\frac{\lambda^-}{p})$ the $p$-distance is given by
\begin{equation}\label{p-distance}
\begin{aligned}
&\mathbb{H}_p(\mathbb{P}^{(\theta^+,\theta^-)} \,|\, \mathbb{P}) = \mathbb{E}_{\mathbb{P}} \left[ \bigg( \frac{d \mathbb{P}^{(\theta^+,\theta^-)}}{d \mathbb{P}} \bigg)^p \right] 
\\ &= e^{-p(\Psi^+(\theta^+) + \Psi^-(\theta^-))} \mathbb{E}_{\mathbb{P}} \big[ e^{p \theta^+ X_1^+} \big] \mathbb{E}_{\mathbb{P}} \big[ e^{p \theta^- X_1^-} \big]
\\ &= \exp \big( -p (\Psi^+(\theta^+) + \Psi^-(\theta^-)) + \Psi^+(p\theta^+) + \Psi^-(p \theta^-) \big)
\\ &= \exp \Big( - \alpha^+ \Gamma(-\beta^+) \Big[ p \big[ (\lambda^+ - \theta^+)^{\beta^+} - (\lambda^+)^{\beta^+} \big] - \big[ (\lambda^+ - p\theta^+)^{\beta^+} - (\lambda^+)^{\beta^+} \big] \Big]
\\ &\quad\quad\quad\quad - \alpha^- \Gamma(-\beta^-) \Big[ p \big[ (\lambda^- - \theta^-)^{\beta^-} - (\lambda^-)^{\beta^-} \big] - \big[ (\lambda^- - p \theta^-)^{\beta^-} - (\lambda^-)^{\beta^-} \big] \Big] \Big).
\end{aligned}
\end{equation}
A similar argumentation as in Theorem~\ref{thm-min-entropy-bilateral} shows that, provided condition (\ref{ineqn-rq-2}) holds true, there exists a pair $(\theta^+,\theta^-)$ minimizing the $p$-distance (\ref{p-distance}), and in this case we also have $\theta^- = \Phi(\theta^+)$, where $\theta^+$ minimizes the function
\begin{align}\label{def-f-p}
\theta \mapsto \mathbb{H}_p(\mathbb{P}^{(\theta,\Phi(\theta))} \,|\, \mathbb{P}).
\end{align}
Numerical computations for concrete examples suggest that $\theta_p \rightarrow \theta_1$ for $p \downarrow 1$, where for each $p > 1$ the parameter $\theta_p$ minimizes (\ref{def-f-p}), and $\theta_1$ minimizes
\begin{align}
\theta \mapsto \mathbb{H}(\mathbb{P}^{(\theta,\Phi(\theta))} \,|\, \mathbb{P}).
\end{align}
This is not surprising, since it is known that, under suitable technical conditions, the $p$-optimal martingale measure converges to the minimal entropy martingale measure for $p \downarrow 1$, see, e.g. \cite{Grandits, Grandits-R, Santacroce-2005, Jeanblanc, Bender, Kohlmann-2008}.

\section{Existence of F\"{o}llmer Schweizer minimal martingale measures}\label{sec-MMM}

In this section, we deal with the existence of the F\"{o}llmer Schweizer minimal martingale measure in tempered stable stock models. This measure has been introduced in \cite{FS} with the motivation of constructing optimal hedging strategies. 
Throughout this section, we fix a finite time horizon $T > 0$ and assume that $\lambda^+ \geq 2$. Then the constant
\begin{align}\label{def-c}
c = c(\alpha^+,\alpha^-,\beta^+,\beta^-,\lambda^+,\lambda^-,r,q) = \frac{\Psi(1) - (r-q)}{\Psi(2) - 2\Psi(1)},
\end{align}
is well-defined. For technical reasons, we shall also assume that the filtration $(\mathcal{F}_t)_{t \geq 0}$ is generated by the tempered stable process of the form (\ref{X-TS}).
As in \cite[Lemma~7.1]{Kuechler-Tappe-pricing}, we show that the discounted stock price process $\tilde{S}$ is a special semimartingale. Let $\tilde{S} = S_0 + M + A$ be its canonical decomposition and let $\hat{Z}$ be the stochastic exponential
\begin{align}\label{density-proc}
\hat{Z}_t = \mathcal{E} \bigg( - \int_0^{\bullet} \frac{c}{\tilde{S}_{s-}} dM_s \bigg)_t, \quad t \in [0,T],
\end{align}
where we recall that for a semimartingale $X$ the stochastic exponential $Y = \mathcal{E}(X)$ defined as
\begin{align*}
\mathcal{E}(X)_t := \exp \bigg( X_t - X_0 - \frac{1}{2} \langle X^c,X^c \rangle \bigg) \prod_{s \leq t} (1 + \Delta X_s) e^{- \Delta X_s}, \quad t \geq 0
\end{align*}
is the unique solution of the stochastic differential equation
\begin{align*}
dY_t = Y_{t-} dX_t, \quad Y_0 = 1,
\end{align*}
see, e.g. \cite[Theorem~I.4.61]{Jacod-Shiryaev}. The (possibly signed) measure $\hat{\mathbb{P}}$ with density
\begin{align}\label{MMM-trafo}
\frac{d \hat{\mathbb{P}}}{d \mathbb{P}} := \hat{Z}_T
\end{align}
is the so-called \textit{F\"{o}llmer Schweizer minimal martingale measure} (in short, \textit{FS minimal martingale measure}).

\begin{theorem}\label{thm-mmm}
The following statements are equivalent:
\begin{enumerate}
\item $\hat{Z}$ is a strict martingale density for $\tilde{S}$.

\item $\hat{Z}$ is a strictly positive $\mathbb{P}$-martingale.

\item We have
\begin{align}\label{cond-Psi-1-2}
-1 \leq c \leq 0.
\end{align}
\item We have
\begin{align}\label{mmm-cond-1}
&\alpha^+ \Gamma(-\beta^+) [ (\lambda^+ - 1)^{\beta^+} - (\lambda^+)^{\beta^+} ] 
\\ \notag &\quad + \alpha^- \Gamma(-\beta) [ (\lambda^- + 1)^{\beta^-} - (\lambda^-)^{\beta^-} ] \leq r-q
\\ \label{mmm-cond-2} \text{and} \quad &\alpha^+ \Gamma(-\beta^+) [ (\lambda^+ - 1)^{\beta^+} - (\lambda^+ - 2)^{\beta^+} ] 
\\ \notag &\quad + \alpha^- \Gamma(-\beta) [ (\lambda^- + 1)^{\beta^-} - (\lambda^- + 2)^{\beta^-} ] \leq -(r-q).
\end{align}
\end{enumerate}
If the previous conditions are satisfied, then under the FS minimal martingale measure $\hat{\mathbb{P}}$ we have
\begin{equation}\label{X-under-MMM}
\begin{aligned}
X &\sim {\rm TS}((c+1)\alpha^+,\beta^+,\lambda^+;(c+1)\alpha^-,\beta^-,\lambda^-)
\\ &\quad * {\rm TS}(-c \alpha^+,\beta^+,\lambda^+ - 1;-c \alpha^-,\beta^-,\lambda^- + 1).
\end{aligned}
\end{equation}
\end{theorem}

\begin{proof}
We only have to show the equivalence (3) $\Leftrightarrow$ (4), as the rest follows by arguing as in the proof of \cite[Theorem~7.3]{Kuechler-Tappe-pricing}. We observe that (\ref{cond-Psi-1-2}) is equivalent to the two conditions
\begin{align*}
\Psi(1) \leq r-q \quad \text{and} \quad \Psi(1) - \Psi(2) \leq - (r-q),
\end{align*}
and, in view of the cumulant generating function given by (\ref{cumulant}), these two conditions are fulfilled if and only if we have (\ref{mmm-cond-1}) and (\ref{mmm-cond-2}).
\end{proof}

\begin{remark}
Relation (\ref{X-under-MMM}) means that under $\hat{\mathbb{P}}$ the driving process $X$ is the sum of two independent tempered stable processes. There are the following two boundary values:
\begin{itemize}
\item In the case $c = 0$ we have
\begin{align*}
X \sim {\rm TS}(\alpha^+,\beta^+,\lambda^+,\alpha^-,\beta^-,\lambda^-) \quad \text{under $\hat{\mathbb{P}}$,}
\end{align*}
i.e., the FS minimal martingale measure $\hat{\mathbb{P}}$ coincides with the physical measure $\mathbb{P}$. Indeed, the
definition (\ref{def-c}) of $c$ and Lemma~\ref{lemma-TS-mm} show that $\mathbb{P}$ already is a martingale measure for $\tilde{S}$.

\item In the case $c = -1$ we have
\begin{align*}
X \sim {\rm TS}(\alpha^+,\beta^+,\lambda^+-1,\alpha^-,\beta^-,\lambda^- + 1) \quad \text{under $\hat{\mathbb{P}}$,}
\end{align*}
i.e., the FS minimal martingale measure $\hat{\mathbb{P}}$ coincides with the Esscher transform $\mathbb{P}^1$, see Theorem~\ref{thm-Esscher}. Indeed, the definition (\ref{def-c}) of $c$ shows that equation (\ref{Esscher-equation}) is satisfied with $\Theta = 1$.
\end{itemize}
\end{remark}

As outlined at the end of \cite[Section~7]{Kuechler-Tappe-pricing}, under the FS minimal martingale measure $\hat{\mathbb{P}}$ we can construct a trading strategy $\xi$ which minimizes the quadratic hedging error. The arguments transfer to our present situation with a driving tempered stable process.

\section{Option pricing in tempered stable stock
models}\label{sec-option}

In this section, we present pricing formulae for European call options. After performing a measure change $\mathbb{Q} \sim \mathbb{P}$ as in Section~\ref{sec-Esscher} or \ref{sec-bilateral-Esscher}, that is, $\mathbb{Q} = \mathbb{P}^{\Theta}$ or $\mathbb{Q} = \mathbb{P}^{(\theta,\Phi(\theta))}$ for appropriate parameters, we may assume that the driving process $X$ is a tempered stable process of the form (\ref{X-TS}) under the martingale measure $\mathbb{Q}$.
We fix a strike price $K > 0$ and a maturity date $T > 0$. Then the price of a European call option with these parameters is given by
\begin{align*}
\pi = e^{-rT} \mathbb{E}_{\mathbb{Q}}[(S_T - K)^+].
\end{align*}
First, we shall derive an option pricing formula in closed form by following an idea from \cite[Section~8.1]{Teichmann}. In the sequel,
\begin{align*}
F_{\alpha^+,\beta^+,\lambda^+;\alpha^-,\beta^-,\lambda^-}
\end{align*}
denotes the ${\rm TS}(\alpha^+,\beta^+,\lambda^+;\alpha^-,\beta^-,\lambda^-)$-distribution function and
\begin{align*}
\bar{F}_{\alpha^+,\beta^+,\lambda^+;\alpha^-,\beta^-,\lambda^-} := 1 - F_{\alpha^+,\beta^+,\lambda^+;\alpha^-,\beta^-,\lambda^-}.
\end{align*}

\begin{proposition}
Suppose that $\lambda^+ > 1$. Then, the price of the call option is given by
\begin{equation}\label{option-price}
\begin{aligned}
\pi &= S_0 e^{(\Psi(1)-r) T} \bar{F}_{\alpha^+ T,\beta^+,\lambda^+ - 1;\alpha^- T,\beta^-,\lambda^- + 1}(\ln (K / S_0)) 
\\ &\quad - e^{-rT} K \bar{F}_{\alpha^+ T,\beta^+,\lambda^+;\alpha^-,\beta^- T,\lambda^-}(\ln (K / S_0)).
\end{aligned}
\end{equation}
\end{proposition}

\begin{proof}
By the definition of the likelihood process (\ref{density-process}) we obtain
\begin{align*}
\pi &= e^{-rT} \mathbb{E}_{\mathbb{Q}}[ (S_T - K)^+ ] = e^{-rT} \mathbb{E}_{\mathbb{Q}}[(S_0 e^{X_T} - K)^+] 
\\ &= S_0 e^{-rT} \mathbb{E}_{\mathbb{Q}} [ e^{X_T} \mathbbm{1}_{\{ X_T \geq \ln (K / S_0) \}} ] - e^{-rT} K \mathbb{Q}(X_T \geq \ln (K / S_0))
\\ &= S_0 e^{-rT} \mathbb{E}_{\mathbb{Q}^1} \Bigg[ e^{X_T} \mathbbm{1}_{\{ X_T \geq \ln (K/S_0) \}} \frac{d \mathbb{Q}}{d \mathbb{Q}^1} \bigg|_{\mathcal{F}_T} \Bigg] - e^{-rT} K \mathbb{Q}(X_T \geq \ln (K/S_0))
\\ &= S_0 e^{-rT} e^{\Psi(1) T} \mathbb{Q}^1 ( X_T \geq \ln (K/S_0) ) - e^{-rT} K \mathbb{Q}(X_T \geq \ln (K/S_0)),
\end{align*}
which, in view of Lemma~\ref{lemma-TS-change-Esscher}, provides the formula (\ref{option-price}).
\end{proof}

\begin{remark}
Note that applying the option pricing formula (\ref{option-price}) requires knowledge about the densities of tempered stable distributions, which are generally not available in closed form. Therefore, we will turn to the option pricing formula (\ref{option-Carr-Madan}) below, which is based on Fourier transform techniques. However, we remark that formula (\ref{option-price}) also holds true in the bilateral Gamma case $\beta^+ = \beta^- = 0$, for which the densities are given in terms of the Whittaker function, see \cite[Section~4]{Kuechler-Tappe}.
\end{remark}

In the sequel, we will use the following option pricing formula (\ref{option-Carr-Madan}), which is based on Fourier transform techniques. Let $X$ is a tempered stable process of the form (\ref{X-TS}) under $\mathbb{P}$, let $\mathbb{Q} \sim \mathbb{P}$ be a martingale measure as in Section~\ref{sec-Esscher}, \ref{sec-bilateral-Esscher} or \ref{sec-MMM}, that is, $\mathbb{Q} = \mathbb{P}^{\Theta}$, $\mathbb{Q} = \mathbb{P}^{(\theta,\Phi(\theta))}$ or $\mathbb{Q} = \hat{\mathbb{P}}$, and denote by $\varphi_{X_T}$ the characteristic function of $X_T$ under the martingale measure $\mathbb{Q}$.

\begin{proposition}\label{prop-formula}
We suppose that
\begin{itemize}
\item $\lambda^+ > 1$, if we have (\ref{X-TS}) under $\mathbb{Q}$. In this case, let $\nu \in (1,\lambda^+)$ be arbitrary.

\item $\lambda^+ > 2$, if we have (\ref{X-under-MMM}) under $\mathbb{Q}$. In this case, let $\nu \in (1,\lambda^+ - 1)$ be arbitrary.
\end{itemize}
Then the price of the call option is given by
\begin{align}\label{option-Carr-Madan}
\pi = - \frac{e^{-rT} K}{2 \pi} \int_{i \nu - \infty}^{i \nu + \infty} \bigg( \frac{K}{S_0} \bigg)^{iz} \frac{\varphi_{X_T}(-z)}{z(z-i)} dz.
\end{align}
\end{proposition}

\begin{proof}
The stock prices are given by
\begin{align*}
S_t = S_0 \exp \big( (r-q)t + \tilde{X}_t \big), \quad t \geq 0,
\end{align*}
where $\tilde{X}$ denotes the L\'{e}vy process given by $\tilde{X}_t = X_t - (r-q) t$ for $t \geq 0$. Moreover, the Fourier transform of the payoff function $w(x) = (e^x - K)^+$ is given by
\begin{align*}
\hat{w}(z) = - \frac{K^{iz + 1}}{z(z-i)}, \quad z \in \mathbb{C} \text{ with } {\rm Im} \, z > 1,
\end{align*}
see Table 3.1 in \cite{Lewis}. Furthermore, the characteristic function (\ref{cf-tempered-stable}) is analytic on the strip $\{ z \in \mathbb{C} : {\rm Im} \, z \in (-\lambda^+,\lambda^-) \}$ by the analyticity of the power function $z \mapsto z^{\beta}$ on the main branch of the complex logarithm for $\beta \in (0,1)$. By our parameter restriction on $\lambda^+$, and possibly taking into account (\ref{X-under-MMM}), we deduce that $\varphi_{X_T}$ is analytic on a strip of the form $\{ z \in \mathbb{C} : {\rm Im} \, z \in (a,b) \}$ with $a < -1$ and $b > 0$. Consequently, \cite[Theorem~3.2]{Lewis} applies and provides us with the call option price
\begin{align*}
\pi &= \frac{e^{-rT}}{2 \pi} \int_{i \nu - \infty}^{i \nu + \infty} e^{-iz (\ln S_0 + (r-q)T)} \varphi_{\tilde{X}_T}(-z) \hat{w}(z) dz
\\ &= - \frac{e^{-rT}}{2 \pi} \int_{i \nu - \infty}^{i \nu + \infty} S_0^{-iz} e^{-iz (r-q)T} \varphi_{X_T}(-z) e^{iz(r-q)T} \frac{K^{iz + 1}}{z(z-i)} dz
\\ &= - \frac{e^{-rT} K}{2 \pi} \int_{i \nu - \infty}^{i \nu + \infty} \bigg( \frac{K}{S_0} \bigg)^{iz} \frac{\varphi_{X_T}(-z)}{z(z-i)} dz,
\end{align*}
which proves (\ref{option-Carr-Madan}). 
\end{proof}

\begin{remark}
Proposition~\ref{prop-formula} does not apply for the boundary case $\lambda^+ = 1$ (or $\lambda^+ = 2$, respectively), although $\mathbb{Q}$ might be a martingale measure in this situation, cf. Lemma~\ref{lemma-TS-mm}. The point is that in this case the characteristic function $\varphi_{X_T}$ is not analytic on a strip of the form $\{ z \in \mathbb{C} : {\rm Im} \, z \in (a,b) \}$ with $a < -1$ and $b > 0$, and hence, the option pricing formula from \cite{Lewis} does not apply. 
\end{remark}

\begin{remark}
The option pricing formula (\ref{option-Carr-Madan}) can also be derived from \cite[Section~3.1]{Carr-Madan}. 
\end{remark}

Taking into account the characteristic function (\ref{cf-tempered-stable}) and relation (\ref{distribution-TS-time}), applying Proposition~\ref{prop-formula} yields the following pricing formulae:
\begin{itemize}
\item Performing the Esscher transform from Section~\ref{sec-Esscher}, for an arbitrary $\nu \in (1,\lambda^+ - \Theta)$ we obtain
\begin{equation}\label{pricing-Esscher}
\begin{aligned}
\pi &= - \frac{e^{-rT} K}{2 \pi} \int_{i \nu - \infty}^{i \nu +
\infty} \left( \frac{K}{S_0} \right)^{iz} \exp \Big( \alpha^+ T \, \Gamma(-\beta^+) \big[ (\lambda^+ - \Theta + iz)^{\beta^+} - (\lambda^+ - \Theta)^{\beta^+} \big] 
\\ &\quad\quad\quad\quad\quad\quad\quad + \alpha^- T \, \Gamma(-\beta^-) \big[ (\lambda^- + \Theta - iz)^{\beta^-} - (\lambda^- + \Theta)^{\beta^-} \big] \Big) \frac{1}{z(z-i)} dz.
\end{aligned}
\end{equation}

\item Performing the bilateral Esscher transform from Section~\ref{sec-bilateral-Esscher}, for an arbitrary $\nu \in (1,\lambda^+ - \theta)$ we obtain
\begin{equation}\label{pricing-entropy}
\begin{aligned}
\pi &= - \frac{e^{-rT} K}{2 \pi} \int_{i \nu - \infty}^{i \nu +
\infty} \left( \frac{K}{S_0} \right)^{iz} \exp \Big( \alpha^+ T \, \Gamma(-\beta^+) \big[ (\lambda^+ - \theta + iz)^{\beta^+} - (\lambda^+ - \theta)^{\beta^+} \big] 
\\ &\quad\quad\quad\quad\quad\quad + \alpha^- T \, \Gamma(-\beta^-) \big[ (\lambda^- - \Phi(\theta) - iz)^{\beta^-} - (\lambda^- - \Phi(\theta))^{\beta^-} \big] \Big) \frac{1}{z(z-i)} dz.
\end{aligned}
\end{equation}

\item Performing option pricing under the FS minimal martingale measure from Section~\ref{sec-MMM}, for an arbitrary $\nu \in (1,\lambda^+ - 1)$ we obtain
\begin{equation}\label{pricing-MMM}
\begin{aligned}
\pi &= - \frac{e^{-rT} K}{2 \pi} \int_{i \nu - \infty}^{i \nu +
\infty} \left( \frac{K}{S_0} \right)^{iz} \exp \Big( (c+1) \alpha^+ T \, \Gamma(-\beta^+) \big[ (\lambda^+ + iz)^{\beta^+} - (\lambda^+)^{\beta^+} \big] 
\\ &\quad\quad\quad\quad\quad\quad\quad + (c+1) \alpha^- T \, \Gamma(-\beta^-) \big[ (\lambda^- - iz)^{\beta^-} - (\lambda^-)^{\beta^-} \big]
\\ &\quad\quad\quad\quad\quad\quad\quad - c \alpha^+ T \, \Gamma(-\beta^+) \big[ (\lambda^+ - 1 + iz)^{\beta^+} - (\lambda^+ - 1)^{\beta^+} \big]
\\ &\quad\quad\quad\quad\quad\quad\quad - c \alpha^- T \, \Gamma(-\beta^-) \big[ (\lambda^- + 1 - iz)^{\beta^-} - (\lambda^- + 1)^{\beta^-} \big] \Big) \frac{1}{z(z-i)} dz,
\end{aligned}
\end{equation}
where the constant $c$ is given by (\ref{def-c}).
\end{itemize}

\section{A case study}\label{sec-case-study}

In order to illustrate our previous results, we shall perform a case study in this section. Figure~\ref{fig-dax} shows historical values of the German stock index DAX from January~3, 2011 until December~28, 2012, and the corresponding log returns. These data are available at {\tt http://www.finanzen.net/index/DAX/Historisch}.

\begin{figure}[!ht]
 \centering
 \includegraphics[width=0.4\textwidth]{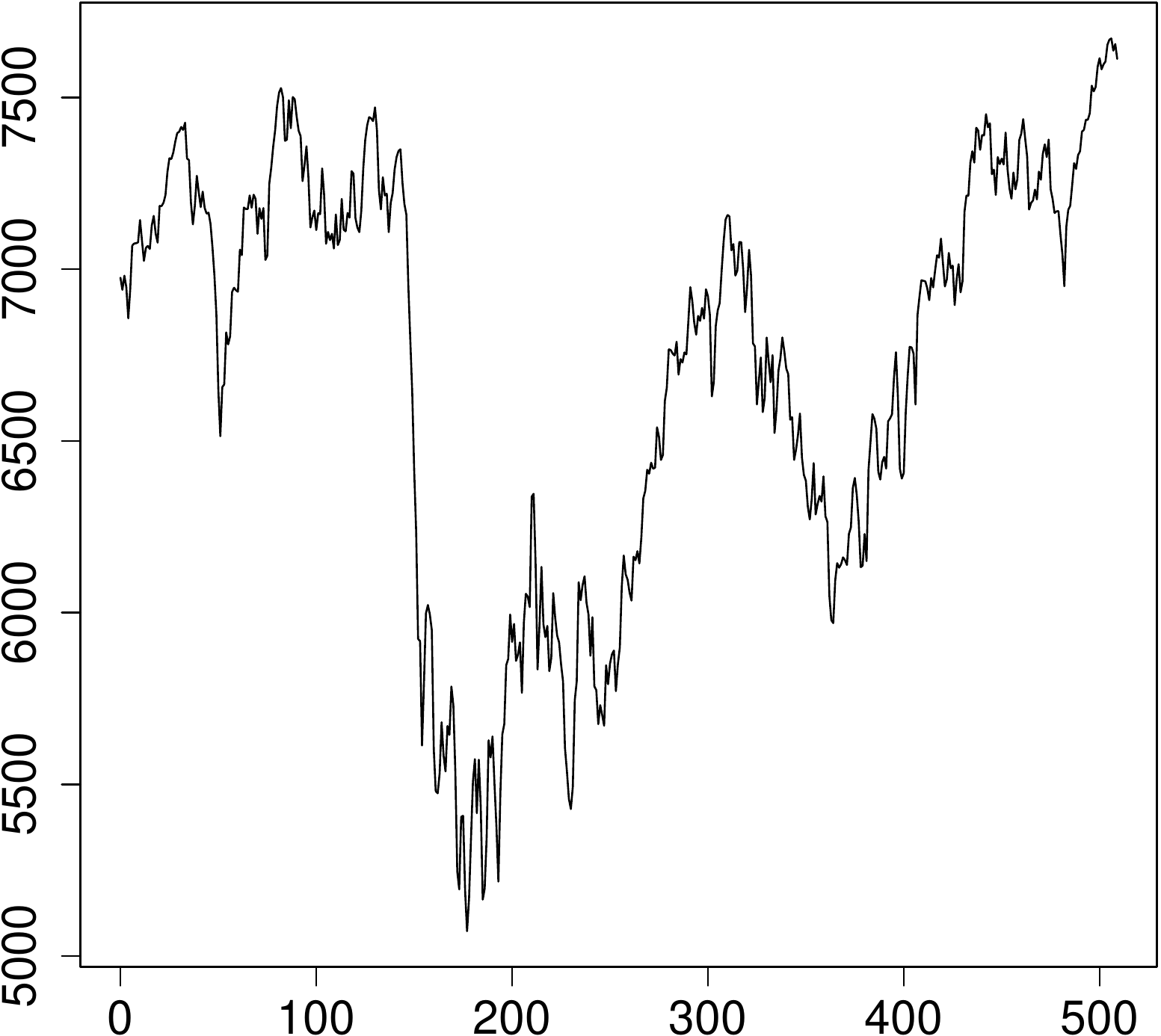}
 \includegraphics[width=0.4\textwidth]{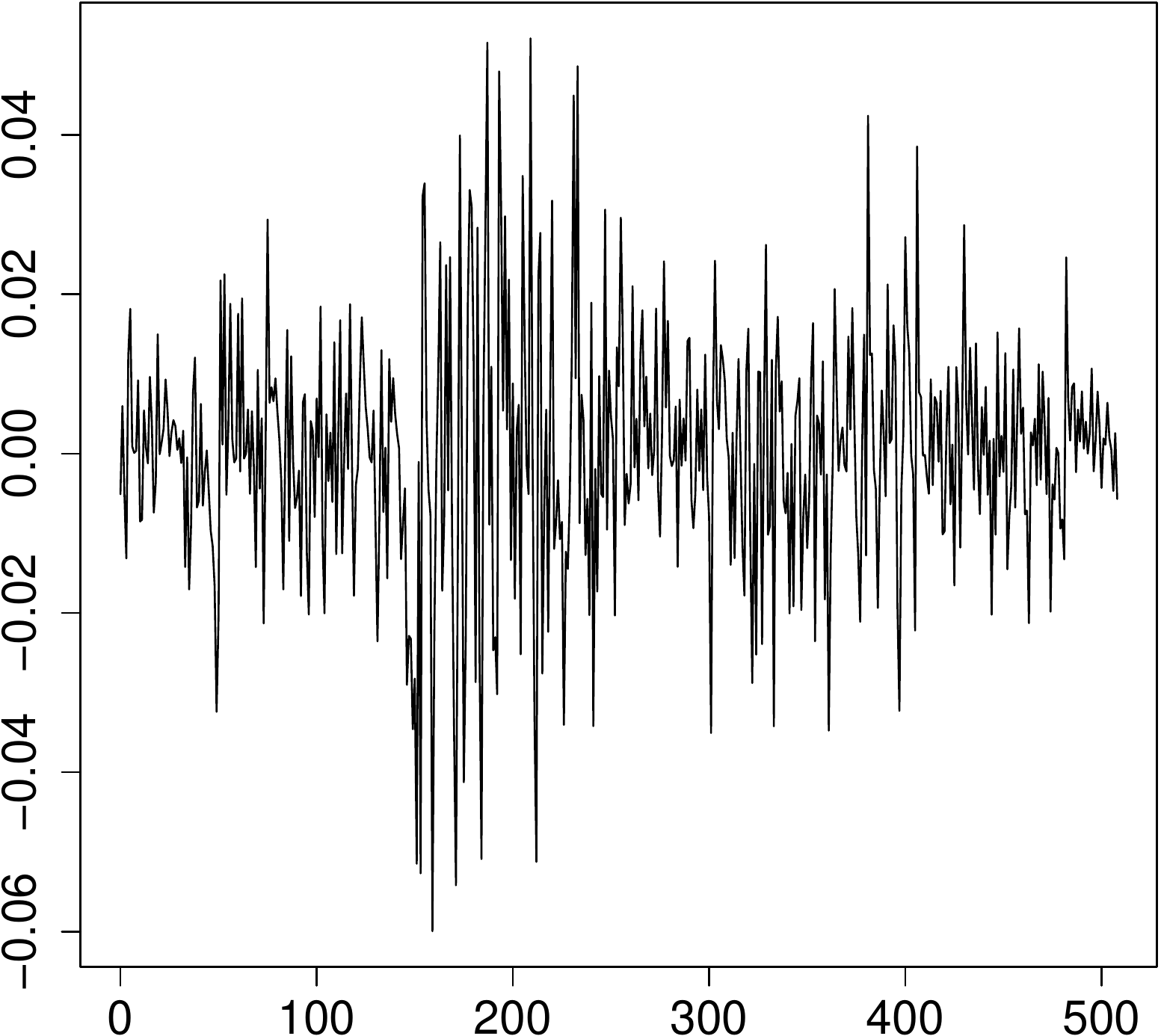}
 \caption{The left plot shows the values of the German stock index DAX from January~3, 2011 until December~28, 2012. The right plot shows the corresponding log returns.}\label{fig-dax}
\end{figure}

In the sequel, the time $t$ is measured in trading days.
The data set consists of $510$ observations, which corresponds to a period of two years. Figure~\ref{fig-hist} below shows a histogram for the log returns. In order to estimate the parameters from these historical data by the method of moments, we determine the empirical moments up to order $4$, which are given by
\begin{align}\label{m1}
m_1 &= 1.717 \cdot 10^{-4},
\\ \label{m2} m_2 &= 2.338 \cdot 10^{-4},
\\ \label{m3} m_3 &= -6.363 \cdot 10^{-7},
\\ \label{m4} m_4 &= 2.716 \cdot 10^{-7}.
\end{align}
Then the parameters with a driving Wiener process $X \sim {\rm N}(\mu,\sigma^2)$ are estimated as
\begin{align}\label{para-Black-Scholes}
\mu = 1.717 \cdot 10^{-4} \quad \text{and} \quad \sigma = 1.529 \cdot 10^{-2}.
\end{align}
The fitted density is shown in the left plot of Figure~\ref{fig-hist}. 

\begin{figure}[!ht]
 \centering
 \includegraphics[width=0.4\textwidth]{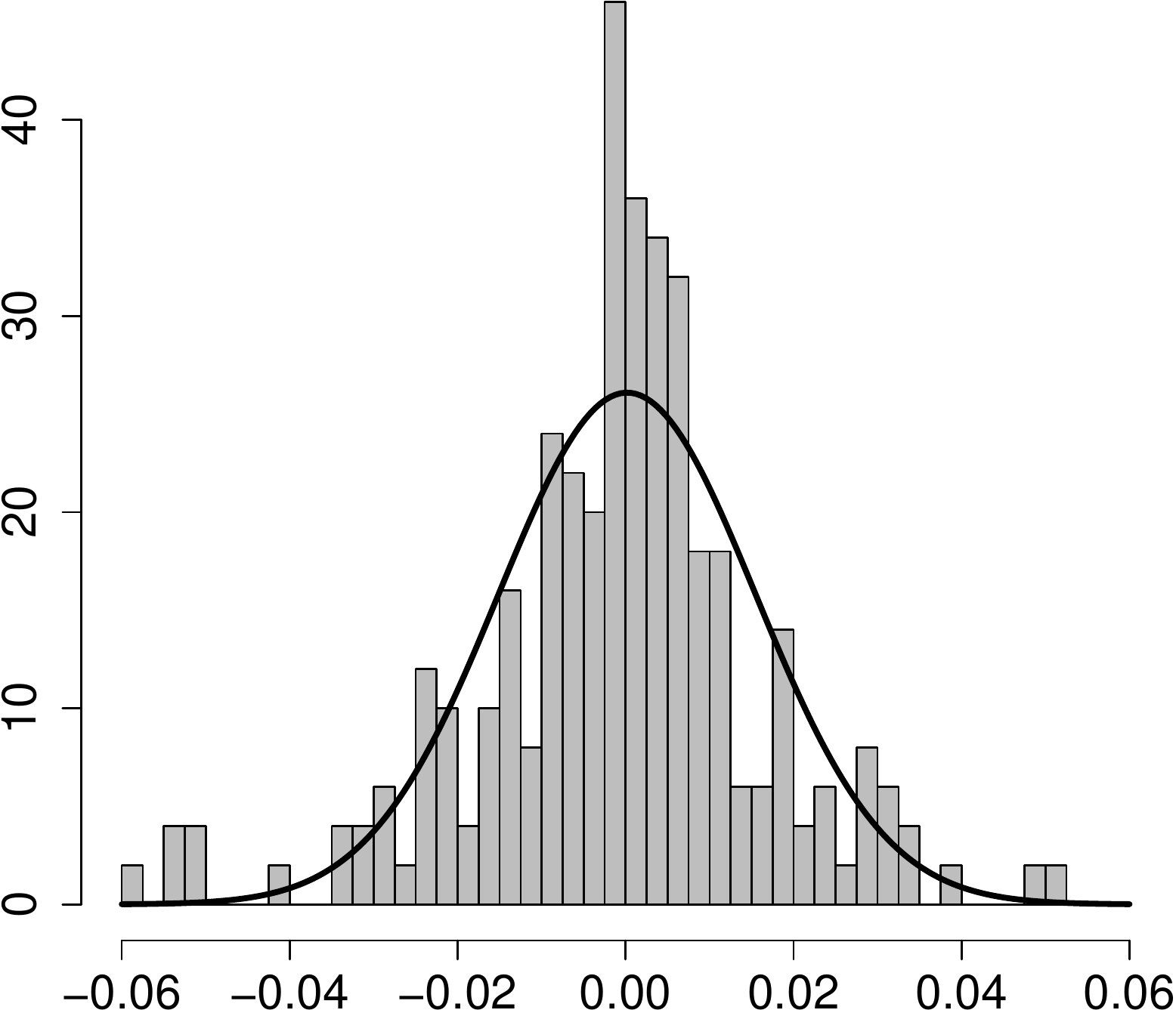}
 \includegraphics[width=0.4\textwidth]{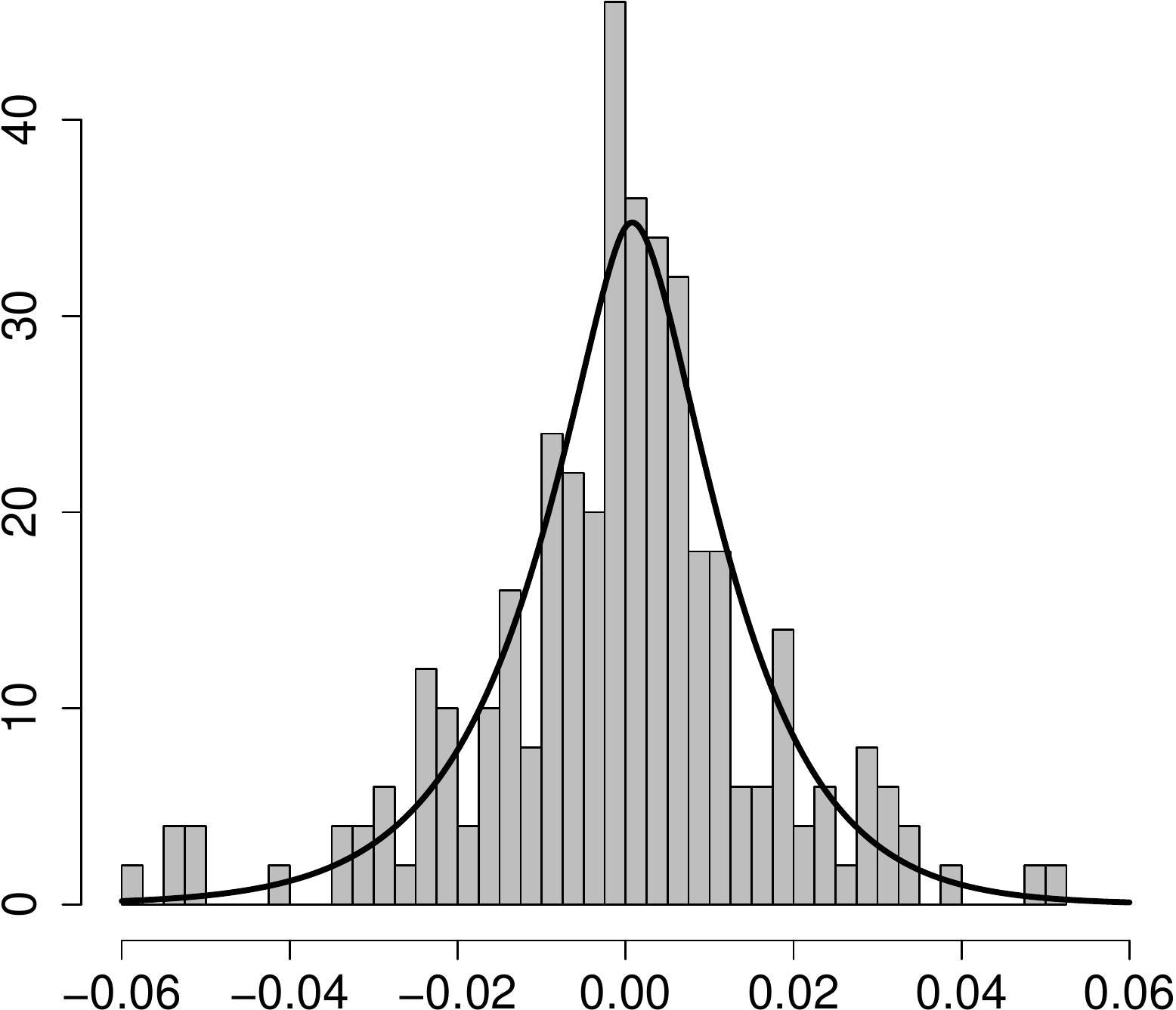}
 \caption{Histogram for the log returns together with the fitted normal distribution in the left plot and the fitted tempered stable distribution in the right plot.}\label{fig-hist}
\end{figure}

It has already been documented in several case studies that the Black Scholes model does not provide a good fit to observed log returns of financial data, and this also shows up here.
Therefore, we consider a tempered stable process $X$ of the form (\ref{X-TS}). Recall that for $\beta^+ = \beta^- = 0$ we would have a bilateral Gamma process. We slightly deviate from this situation by choosing
\begin{align}\label{beta-est}
\beta := \beta^+ = \beta^- = 0.1. 
\end{align}
In order to estimate the remaining parameters by the method of moments, we have to solve the system of equations
\begin{align}\label{system-equations}
\left\{
\begin{array}{rcl}
\alpha^+ (\lambda^-)^{1-\beta} - \alpha^- (\lambda^+)^{1-\beta} - c_1 (\lambda^+)^{1-\beta} (\lambda^-)^{1-\beta} & = & 0
\\ \alpha^+ (\lambda^-)^{2-\beta} + \alpha^- (\lambda^+)^{2-\beta} - c_2 (\lambda^+)^{2-\beta} (\lambda^-)^{2-\beta} & = & 0
\\ \alpha^+ (\lambda^-)^{3-\beta} - \alpha^- (\lambda^+)^{3-\beta} - c_3 (\lambda^+)^{3-\beta} (\lambda^-)^{3-\beta} & = & 0
\\ \alpha^+ (\lambda^-)^{4-\beta} + \alpha^- (\lambda^+)^{4-\beta} - c_4 (\lambda^+)^{4-\beta} (\lambda^-)^{4-\beta} & = & 0,
\end{array}
\right.
\end{align}
where $c_1, c_2, c_3, c_4$ are given by 
\begin{align*}
c_1 &= m_1 / \Gamma(1-\beta),
\\ c_2 &= ( m_2 - m_1^2 ) / \Gamma(2-\beta),
\\ c_3 &= ( m_3 - 3m_1 m_2 + 2m_1^3 ) / \Gamma(3-\beta),
\\ c_4 &= ( m_4 - 4m_1 m_3 - 3m_2^2 + 12 m_1^2 m_2 - 6m_1^4 ) / \Gamma(4-\beta),
\end{align*}
see \cite[Section~6]{Kuechler-Tappe-TS} for further details. The solution of (\ref{system-equations}) is given by
\begin{align}\label{remaining-est}
\alpha^+ = 1.0260, \quad \alpha^- = 0.8506, \quad \lambda^+ = 122.58, \quad \lambda^- = 100.86.
\end{align}
The right plot in Figure~\ref{fig-hist} shows the fitted tempered stable density. Recall that the densities of tempered stable distributions are generally not available in closed form. For the right plot in Figure~\ref{fig-hist} we have used the inversion formula
\begin{align*}
f(x) &= \frac{1}{2\pi} \int_{\mathbb{R}} \exp \Big( -ixz + \alpha^+ \Gamma(-\beta^+) \big[ (\lambda^+ - iz)^{\beta^+} - (\lambda^+)^{\beta^+} \big] 
\\ &\qquad\qquad\qquad\,\, + \alpha^- \Gamma(-\beta^-) \big[ (\lambda^+ + iz)^{\beta^-} - (\lambda^-)^{\beta^-} \big] \Big) dz,
\end{align*}
which follows from (\ref{cf-tempered-stable}) and \cite[Lemma 28.5, Proposition 2.5.xii]{Sato}.

In the sequel, we suppose that under the real-world probability measure $\mathbb{P}$ the process $X$ is a tempered stable process of the form (\ref{X-TS}) with (\ref{beta-est}) and estimated parameters (\ref{remaining-est}). As the time $t$ is measured in trading days, the interest rate $r$ denotes the daily interest rate. We suppose that it is given by $r = 0.01 / 255$, which corresponds to an annualized interest rate $r_a$ of $1\%$. Moreover, we suppose that $q = 0$, i.e., the stock does not pay dividends. 

Based on these data, we will illustrate our results from  Sections~\ref{sec-Esscher}--\ref{sec-MMM} concerning the existence of equivalent martingale measures. 

\begin{figure}[!ht]
 \centering
 \includegraphics[width=0.4\textwidth]{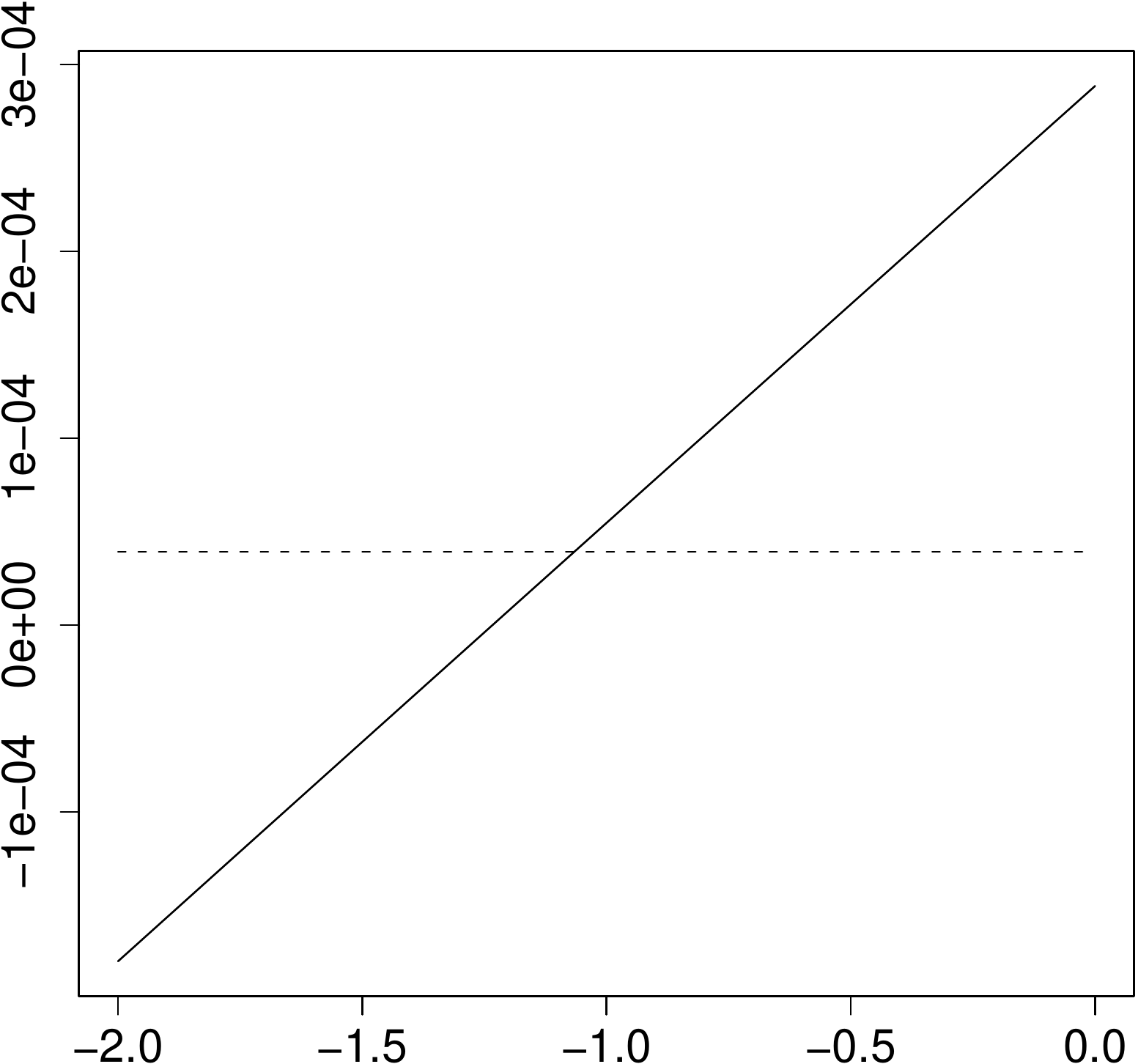}
 \includegraphics[width=0.4\textwidth]{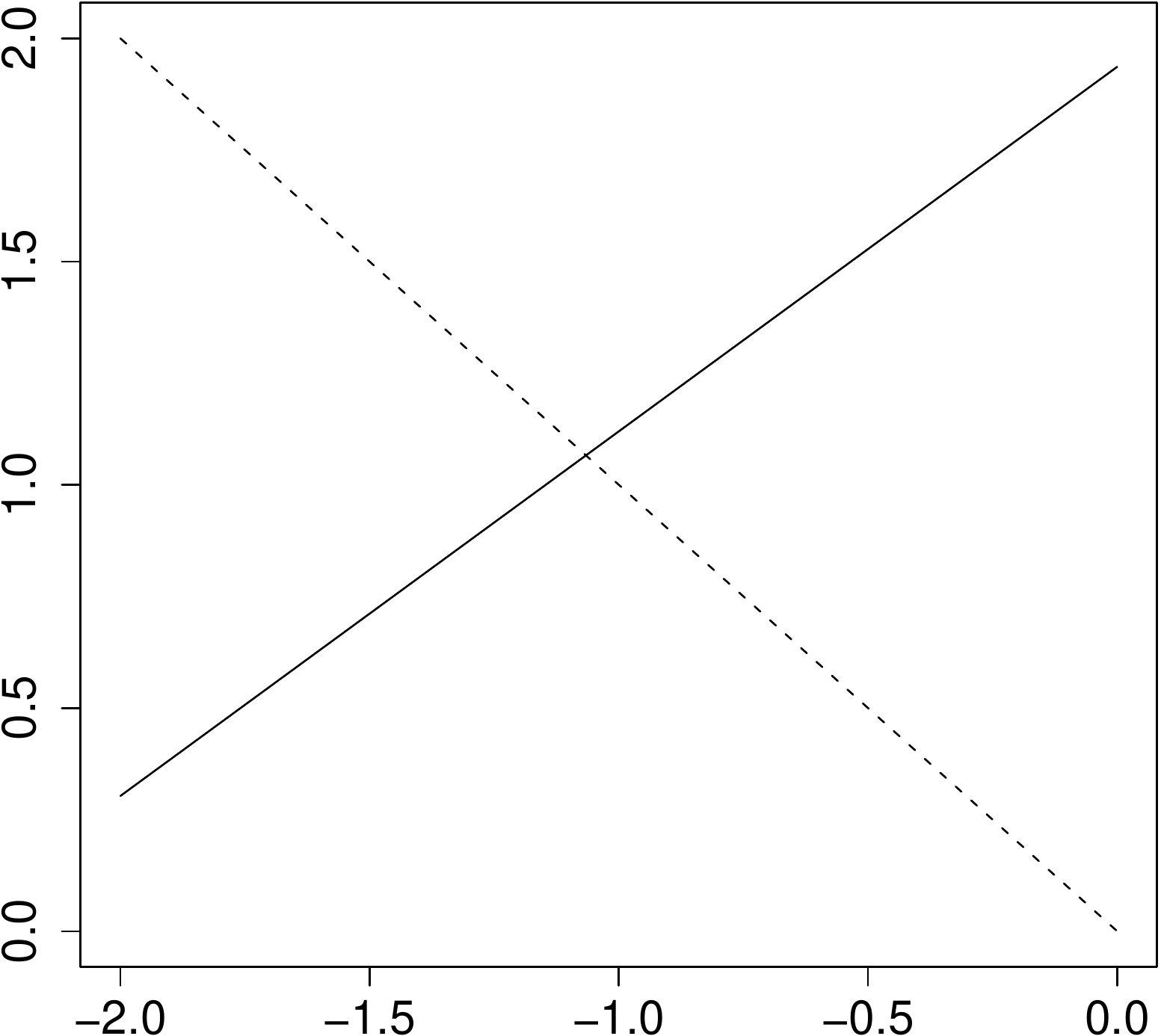}
 \caption{The left plot shows the function $f$ from Section~\ref{sec-Esscher} together with the interest rate $r$ as dashed line. The right plot shows the graph of $\Phi$ together with the graph of $\Theta \mapsto -\Theta$ as dashed line.}\label{fig-Esscher}
\end{figure}

First, we consider the Esscher martingale measure from Section~\ref{sec-Esscher}. The left plot in Figure~\ref{fig-Esscher} shows the function $f : [-\lambda^-,\lambda^+ - 1] \rightarrow \mathbb{R}$ on the interval $[-2,0]$, together with the interest rate $r$ as dashed line. 
As this plot indicates, condition (\ref{cond-domain-temp}) is fulfilled and the solution of equation (\ref{Esscher-equation}) is given by 
\begin{align}
\Theta = -1.0659. 
\end{align}
Therefore, the Esscher martingale measure $\mathbb{P}^{\Theta}$ exists. Alternatively, this can be seen by inspecting the right plot in Figure~\ref{fig-Esscher}, which shows the function $\Phi$ from Proposition~\ref{prop-mm-bilateral} on the interval $[-2,0]$, together with the graph of $\Theta \mapsto -\Theta$ as dashed line. The graph of $\Phi$ represents all martingale measures $\mathbb{P}^{(\theta,\Phi(\theta))}$ which preserve the class of tempered stable processes, and the dashed line represents all Esscher transforms $\mathbb{P}^{\Theta} = \mathbb{P}^{(\Theta,-\Theta)}$. Therefore, the intersection point corresponds to the just determined Esscher martingale measure $\mathbb{P}^{\Theta}$.

\begin{figure}[!ht]
 \centering
 \includegraphics[width=0.4\textwidth]{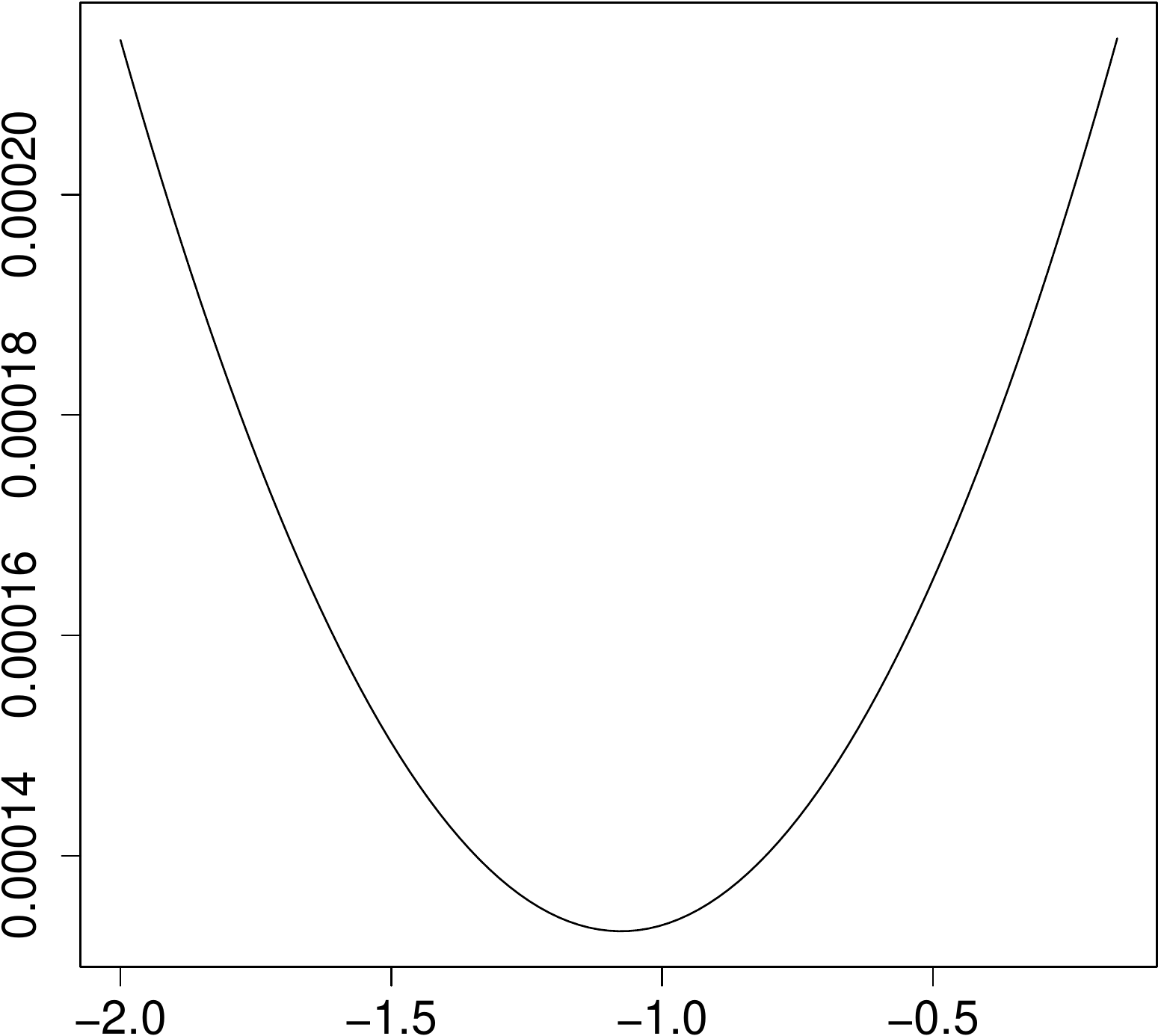}
 \includegraphics[width=0.4\textwidth]{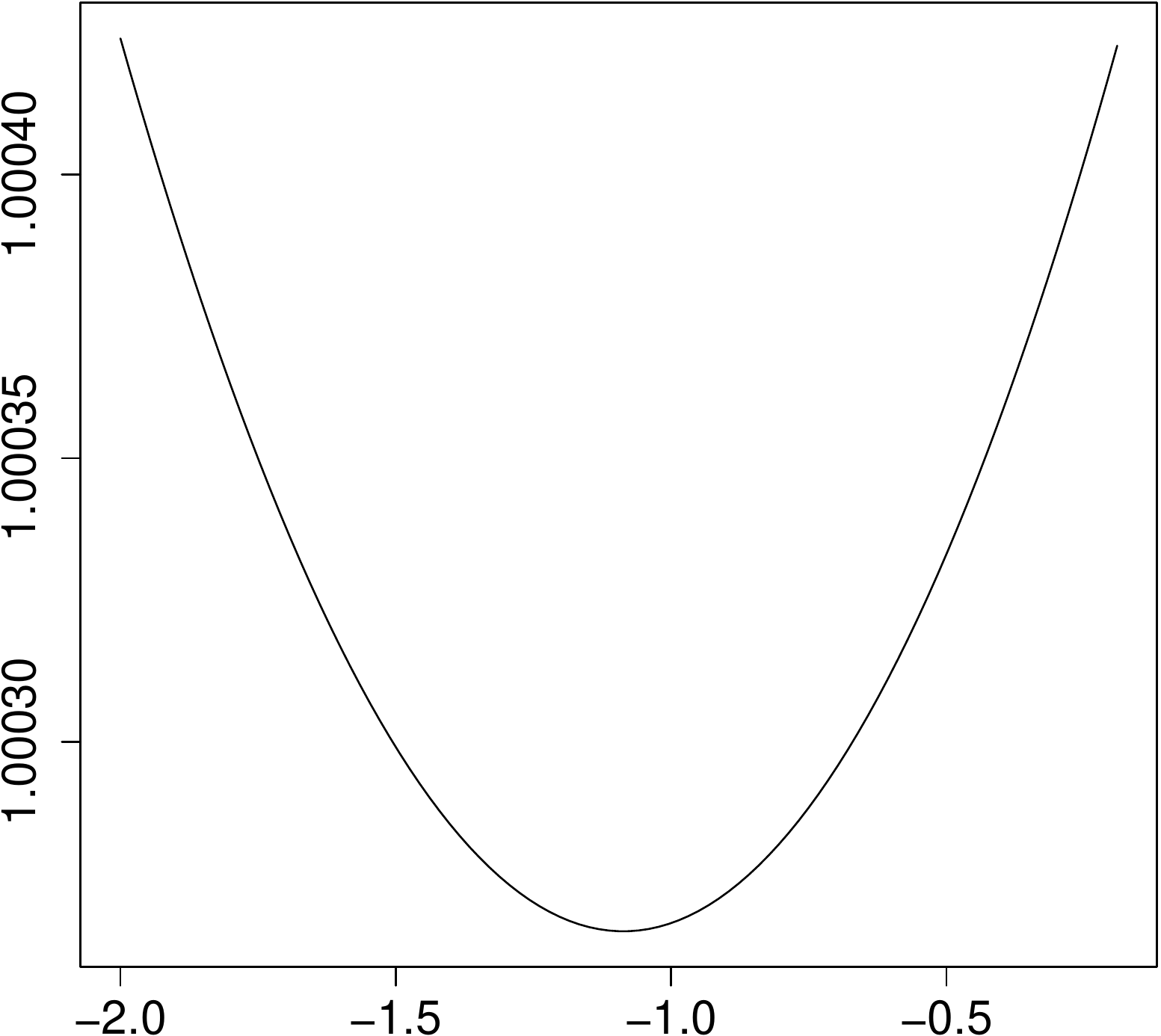}
 \caption{The left plot shows the relative entropies on the interval $[-2,0]$, where the $x$-axis is $\theta$ and the $y$-axis is $\mathbb{H}(\mathbb{P}^{(\theta,\Phi(\theta))} \,|\, \mathbb{P})$. The right plot shows the $2$-distances on the interval $[-2,0]$, where the $x$-axis is $\theta$ and the $y$-axis is $\mathbb{H}_2(\mathbb{P}^{(\theta,\Phi(\theta))} \,|\, \mathbb{P})$.}\label{fig-min}
\end{figure}

Next, we treat the existence of the minimal bilateral Esscher martingale measures from Section~\ref{sec-bilateral-Esscher}. The left plot in Figure~\ref{fig-min} shows the relative entropies $\theta \mapsto \mathbb{H}(\mathbb{P}^{(\theta,\Phi(\theta))} \,|\, \mathbb{P})$ on the interval $[-2,0]$; it indicates that the minimal entropy martingale measure within the class of bilateral Esscher transforms is attained for 
\begin{align}\label{theta-1-est}
\theta_1 = -1.0760. 
\end{align}
The right plot in Figure~\ref{fig-min} shows the $2$-distances $\theta \mapsto \mathbb{H}_2(\mathbb{P}^{(\theta,\Phi(\theta))} \,|\, \mathbb{P})$ on the interval $[-2,0]$; it indicates that the variance-optimal martingale measure within the class of bilateral Esscher transforms is attained for 
\begin{align}\label{theta-2-est}
\theta_2 = -1.0868.
\end{align}

\begin{figure}[!ht]
 \centering
 \includegraphics[width=0.5\textwidth]{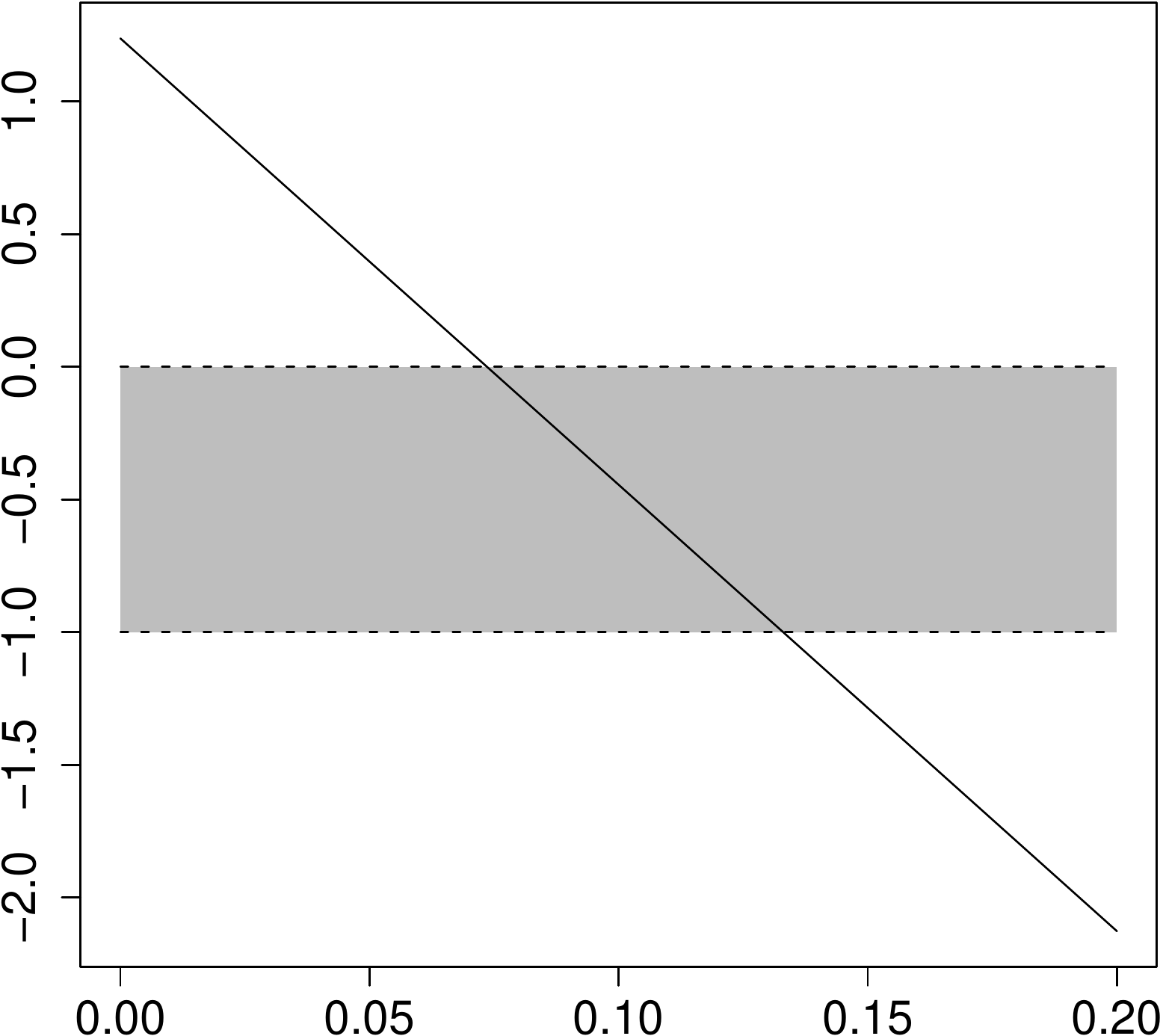}
 \caption{The function $r_a \mapsto c(r_a)$ defined according to (\ref{fct-const}) on the interval $[0,0.2]$, where the $x$-axis is the annualized interest rate $r_a$. The shaded area indicates the values of $r_a$ for which the FS minimal martingale measure exists.}\label{fig-mmm}
\end{figure}

Finally, we treat the existence of the FS minimal martingale measure from Section~\ref{sec-MMM}. For this purpose, it will be useful to consider the annualized interest rate $r_a$. Figure~\ref{fig-mmm} shows the function
\begin{align}\label{fct-const}
r_a \mapsto c(\alpha^+,\alpha^-,\beta^+,\beta^-,\lambda^+,\lambda^-,r_a/255,q)
\end{align}
defined according (\ref{def-c}) with varying annualized interest rate $r_a$ on the interval $[0,0.2]$. According to Theorem~\ref{thm-mmm}, the FS minimal martingale measure exists if and only if $-1 \leq c \leq 0$, that is, the values of $c$ belong to the shaded area in Figure~\ref{fig-mmm}. We see that the FS minimal martingale measure exists if and only if $0.0736 \leq r_a \leq 0.1330$, that is, the annual interest rate is between $7.36\%$ and $13.30\%$. In particular, in our model with an annual interest rate of $1\%$ the FS minimal martingale measure does not exist.

\begin{figure}[!ht]
 \centering
 \includegraphics[width=0.4\textwidth]{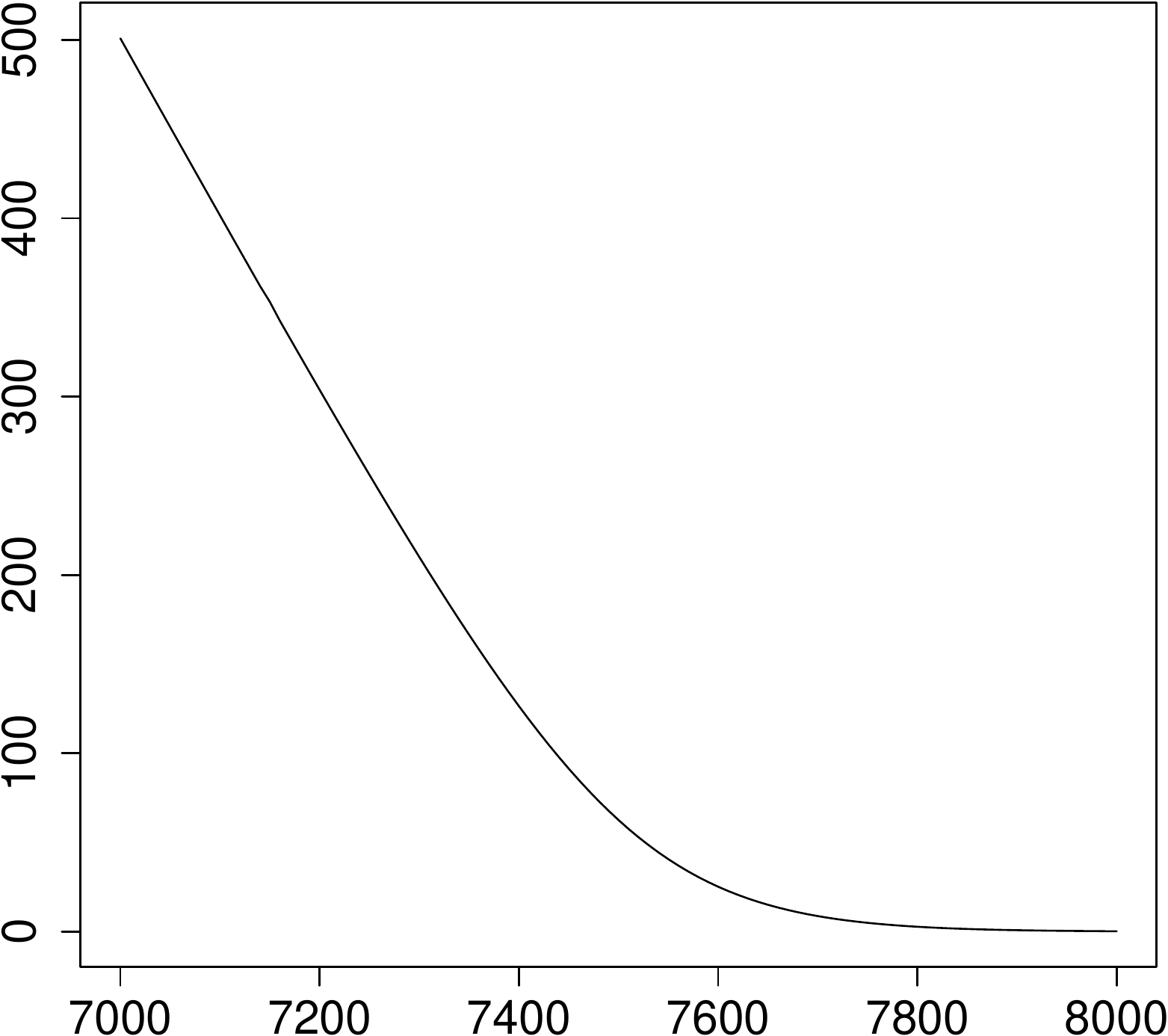}
 \includegraphics[width=0.4\textwidth]{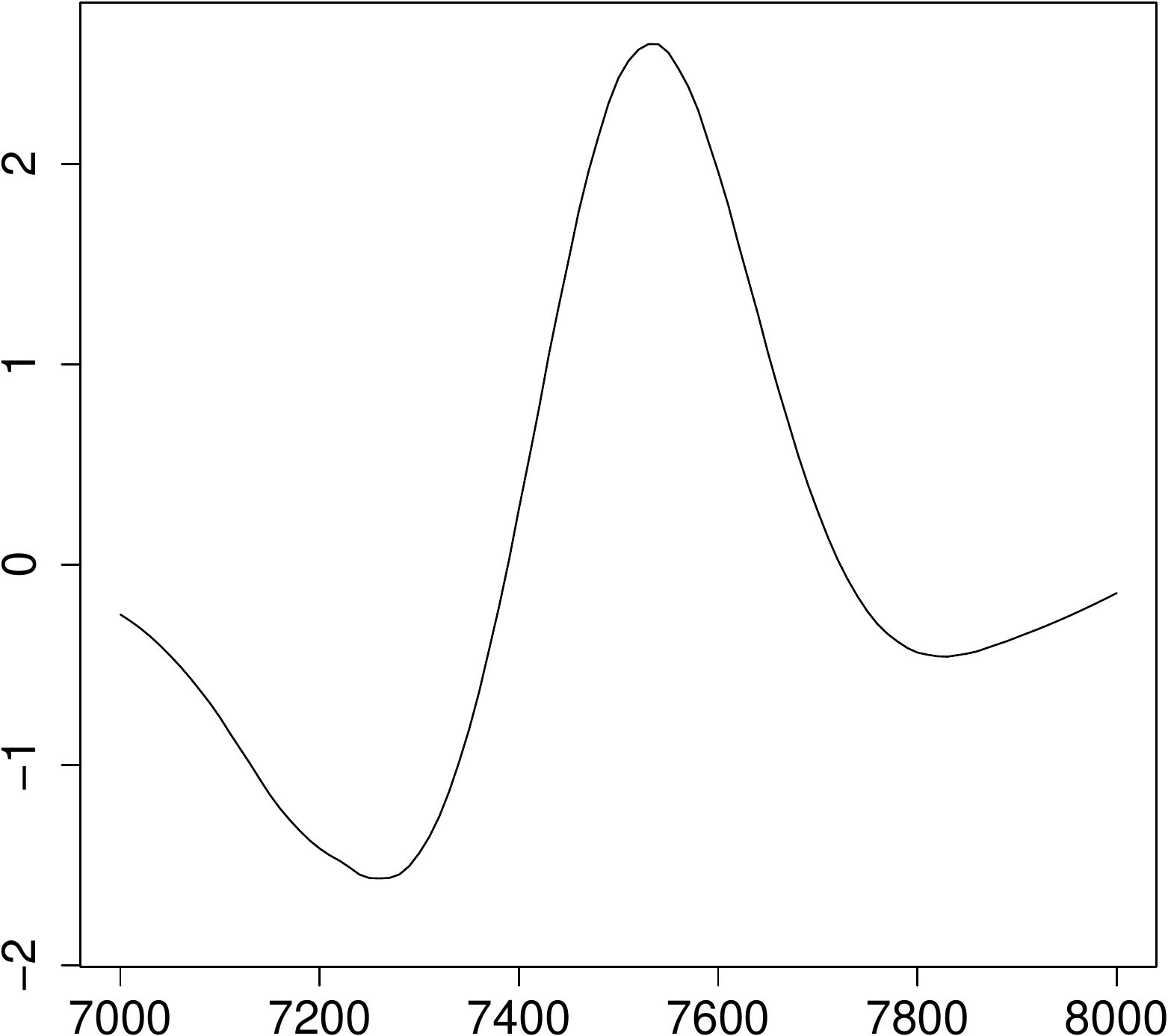}
 \caption{The left plot shows the prices of European call options with $S_0 = 7500$, $T = 2$ and $K \in [7000,8000]$, where the $x$-axis is the strike price $K$, computed with the minimal entropy martingale measure. The right plot shows these prices minus the corresponding Black Scholes prices.}\label{fig-prices}
\end{figure}

In the sequel, we shall illustrate our results from Section~\ref{sec-option} concerning option pricing formulae.

The left plot in Figure~\ref{fig-prices} shows the prices of European call options with current stock price $S_0 = 7500$, date of maturity $T = 2$ and strike prices $K$ varying from $7000$ to $8000$. We have computed these prices with the minimal entropy martingale measure, i.e. with formula (\ref{pricing-entropy}), where $\theta = \theta_1$ is given by (\ref{theta-1-est}), and where the model parameters are given by (\ref{beta-est}) and (\ref{remaining-est}). The right plot in Figure~\ref{fig-prices} shows the difference between these prices and the corresponding Black Scholes prices.

\begin{figure}[!hb]
 \centering
 \includegraphics[width=0.6\textwidth]{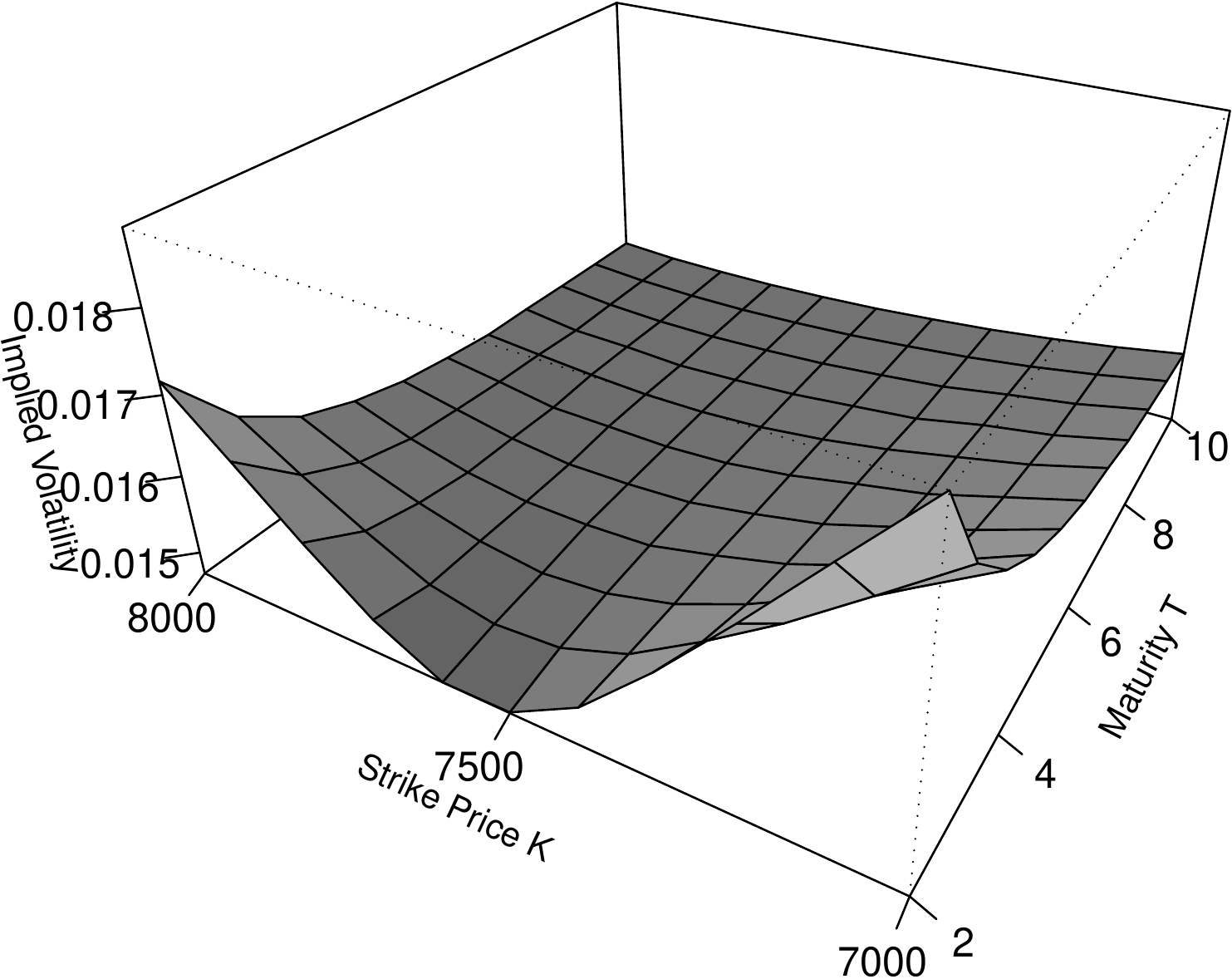}
 \caption{The implied volatility surface computed with the minimal entropy martingale measure. The parameters for the call option are $S_0 = 7500$, $T \in [2,10]$ and $K \in [7000,8000]$.}\label{fig-implied_vol}
\end{figure}

Figure~\ref{fig-implied_vol} shows the implied volatility surface with current stock price $S_0 = 7500$, maturity dates $T$ varying from $2$ to $10$, and strike prices $K$ varying from $7000$ to $8000$. For this procedure, we have computed option prices with the minimal entropy martingale measure, i.e. with formula (\ref{pricing-entropy}), where $\theta = \theta_1$ is given by (\ref{theta-1-est}), and where the model parameters are given by (\ref{beta-est}) and (\ref{remaining-est}), and inverted the Black Scholes formula for the standard deviation $\sigma$.
We observe a volatility smile for $T=2$, which flattens out for longer times of maturity and converges to the standard deviation $\sigma$ of the Black Scholes model, which we have estimated in (\ref{para-Black-Scholes}). For $T=10$ the implied volatility curve shown in Figure~\ref{fig-implied_vol} behaves almost like the constant function which is equal to $\sigma$ estimated in (\ref{para-Black-Scholes}); this flat behaviour does not change for larger times of maturity $T$. Our empirical observation is not surprising, as we have shown in \cite[Theorem~4.10]{Kuechler-Tappe-TS} that, for a tempered stable process $X$ and a Brownian motion $W$ with the same mean and variance, the distributions of $X_t$ and $W_t$ are close to each other for large time points $t$; see also \cite[Theorem~3.1.ii]{Rosinski} for an investigation of the long time behaviour of tempered stable processes.

\begin{figure}[!ht]
 \centering
 \includegraphics[width=0.6\textwidth]{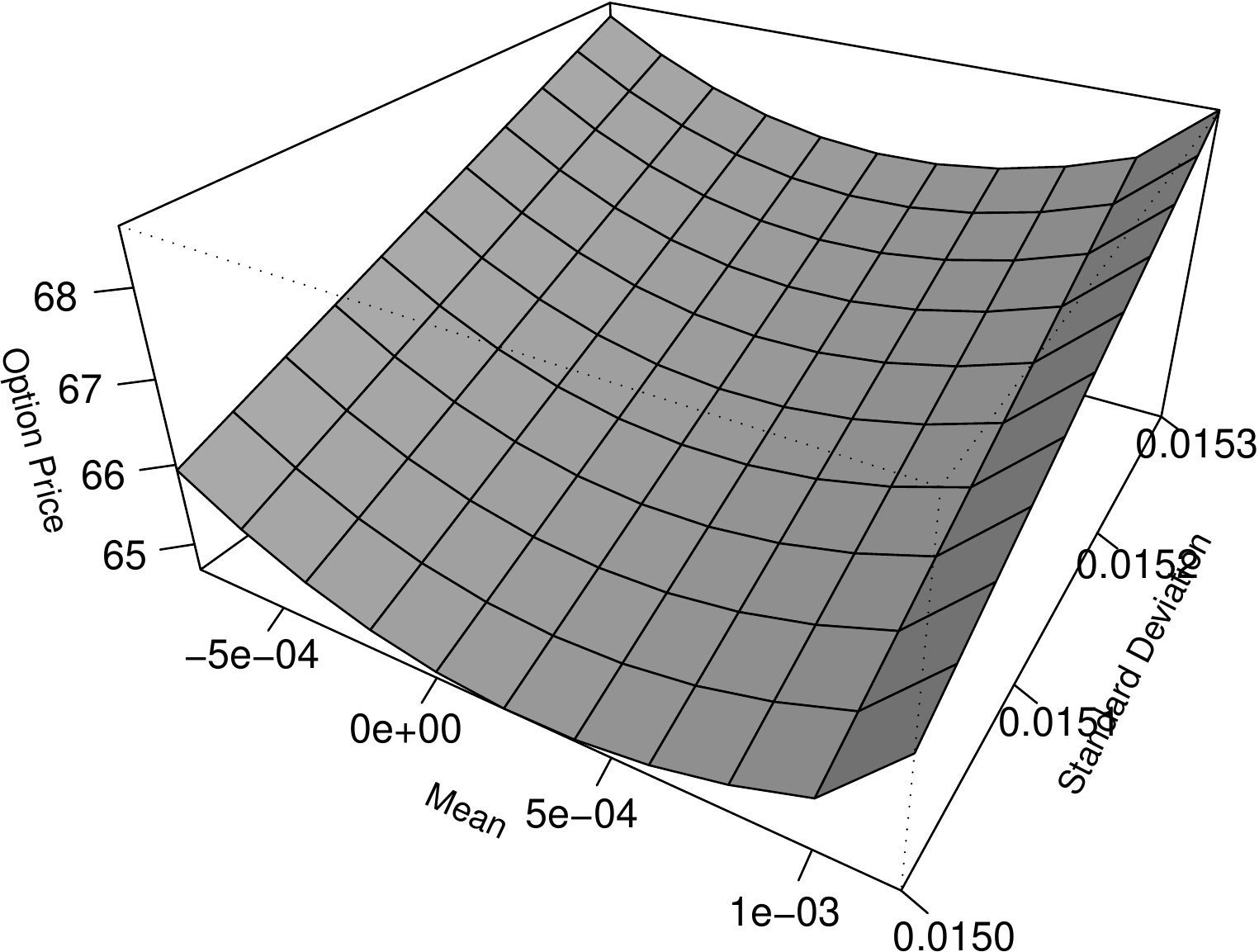}
 \caption{Option prices for the mean $\mu \in [-0.0008,0.0012]$ and the standard deviation $\sigma \in [0.0150,0.0153]$. The parameters for the call option are $S_0 = 7500$, $T = 10$ and $K=7700$.}\label{fig-calibration}
\end{figure}

It is well known that, when estimating the model parameters, reasonable confidence intervals for the mean $\mu$ can only be achieved for a very large number of observations. Therefore, it is important that the model behaves stable with respect to calibration errors. In order to demonstrate the stability of our pricing rules, we have computed option prices for various values of the mean $\mu = m_1$ and the standard deviation $\sigma = \sqrt{m_2 - m_1^2}$. For this procedure, we have calculated the respective parameters $\alpha^+,\alpha^-,\lambda^+,\lambda^- > 0$ by solving the system of equations (\ref{system-equations}) with $m_3, m_4$ given by (\ref{m3}), (\ref{m4}) and $\beta \in (0,1)$ given by (\ref{beta-est}), and computed the option prices with the minimal entropy martingale measure, i.e. with formula (\ref{pricing-entropy}) and $\theta = \theta_1$. Figure~\ref{fig-calibration} shows the computed option prices for $\mu$ varying from $-0.0008$ to $0.0012$, and $\sigma$ varying from $0.0150$ to $0.0153$, with current stock price $S_0 = 7500$, date of maturity $T = 10$ and strike price $K=7700$. The surface behaves locally flat and shows that the model is stable with respect to minor calibration errors.

\section{Conclusion}

In this paper, we have provided a systematic analysis of the existence of equivalent martingale measures for exponential stock price models driven by tempered stable processes, under which the computation of option prices remains analytically tractable.

In this section, we shall review our results and provide a comparison with the results derived in \cite{Bianchi-book}. The textbook \cite{Bianchi-book} deals with financial models driven by several types of tempered stable processes. Its studies encompass the CTS, GTS, KRTS, MTS, NTS, and RDTS processes. We refer to \cite{Bianchi-book} for further details, but point out that the tempered stable distributions considered in this paper correspond to the \emph{generalized classical tempered stable} (GTS) \emph{distributions} with mean
\begin{align}\label{m-mean}
m = \Gamma(1 - \beta^+) \frac{\alpha^+}{(\lambda^+)^{1 - \beta^+}} - \Gamma(1 - \beta^-) \frac{\alpha^-}{(\lambda^-)^{1 - \beta^-}},
\end{align}
see formula (3.4) on page~68 in \cite{Bianchi-book}, which seems to have a small typo. Note that the calculation of the mean in (\ref{m-mean}) is also consistent with formula (2.12) in \cite{Kuechler-Tappe-TS}.

As pointed out in Remark~\ref{rem-measures}, all measure transformations preserving the class of tempered stable processes are bilateral Esscher transforms. This is due to the result that for \begin{align*}
X \sim {\rm TS}(\alpha_1^+,\beta_1^+,\lambda_1^+;\alpha_1^-,\beta_1^-,\lambda_1^-) 
\end{align*}
under a probability measure $\mathbb{P}$ and
\begin{align*}
X \sim {\rm TS}(\alpha_2^+,\beta_2^+,\lambda_2^+;\alpha_2^-,\beta_2^-,\lambda_2^-) 
\end{align*}
under another probability measure $\mathbb{Q}$, the measures $\mathbb{P}$ and $\mathbb{Q}$ are equivalent if and only if $\alpha_1^+ = \alpha_2^+$, $\alpha_1^- = \alpha_2^-$, $\beta_1^+ = \beta_2^+$ and $\beta_1^- = \beta_2^-$. In Section~5.3.3 in \cite{Bianchi-book}, such a result has also been shown for GTS-processes, and the characteristic triplet of the logarithm of the Radon-Nikodym derivative has been determined.

Using bilateral Esscher transforms, we have investigated several martingale measures under which the driving process remains a tempered stable process. These martingale measures have been the Esscher martingale measure in Section~\ref{sec-Esscher} (which later turned out to be a special case of bilateral Esscher martingale measures), and the minimal entropy martingale measure as well as the $p$-optimal martingale measure in Section~\ref{sec-bilateral-Esscher}. Furthermore, we have provided a criterion for the existence of the F\"{o}llmer Schweizer minimal martingale measure in Section~\ref{sec-MMM}. In case of existence, the driving process turned out to be the sum of two independent tempered stable processes under the new measure, thus providing an analytically tractable model.

In Section~\ref{sec-option}, we have provided the option pricing formulae (\ref{pricing-Esscher})--(\ref{pricing-MMM}), which apply to the martingale measures that we have studied in the aforementioned sections. These formulae are based on Fourier transform techniques and follow from a result in \cite{Lewis}, which has also been provided in \cite{Carr-Madan}. An option pricing formula of this kind can also be found in Section~7.5 in \cite{Bianchi-book}; see formula (7.10) on page~152.

In our case study in Section~\ref{sec-case-study}, we have estimated the parameters of the tempered stable process from historical data of the German stock index DAX. Based on our previous results, we have determined appropriate martingale measures and have used these in order to compute option prices and implied volatility surfaces. In Section~7.5.2 in \cite{Bianchi-book}, the authors have proceeded differently. Namely, they do not consider the real-world probability measure, they rather calibrate the risk-neutral parameters of the tempered stable process from available option price data. Thus, they assume that the driving process is also a tempered stable process under the martingale measure, which means that the martingale measure is a bilateral Esscher martingale measure. However, comparing our Figure~\ref{fig-implied_vol} with Figure~7.1 on page~154 in \cite{Bianchi-book}, we observe similar results concerning the implied volatility surfaces: For short maturity dates we have a volatility smile, which flattens out for longer maturities.

The class of bilateral Gamma distributions, which occurs for $\beta^+ = \beta^- = 0$, is a limiting case within the class of tempered stable distributions. Stock price models driven by bilateral Gamma processes have been examined in \cite{Kuechler-Tappe-pricing}. Comparing our results from this paper with those from \cite{Kuechler-Tappe-pricing}, we see that for tempered stable processes with $\beta^+,\beta^- \in (0,1)$ we obtain 
more restrictive conditions concerning the existence of appropriate martingale measures than for driving bilateral Gamma processes. This is not surprising, as our investigations in \cite{Kuechler-Tappe-TS} have shown that, in many respects, the properties of bilateral Gamma distributions differ from those of all other tempered stable distributions.

\section*{Acknowledgement}

The authors are grateful to two anonymous referees for valuable comments and suggestions.


\begin{thebibliography}{20}

\bibitem{Bender}
    Bender, C. and Niethammer, C. (2008):
    On $q$-optimal martingale measures in exponential L\'evy models.
    \emph {Finance and Stochastics} {\bf 12}(3), 381--410.

\bibitem{Bianchi-2}
    Bianchi, M.~L., Rachev, S.~T., Kim, Y.~S. and Fabozzi, F.~J. (2010):
    Tempered stable distributions and processes in finance: numerical analysis.
    \emph {Mathematical and statistical methods for actuarial sciences and finance} 33--42, Springer Italia, Milan.

\bibitem{Bianchi-3}
    Bianchi, M.~L., Rachev, S.~T., Kim, Y.~S. and Fabozzi, F.~J. (2011):
    Tempered infinitely divisible distributions and processes.
    \emph {Theory of Probability and its Applications } {\bf 55}(1), 2--26.

\bibitem{Boy}
    Boyarchenko, S.~I. and Levendorskii, S.~Z. (2000):
    Option pricing for truncated L\'{e}vy processes. 
    \emph {International Journal of Theoretical and Applied Finance} {\bf 3}(3), 549--552.

\bibitem{Carr-Madan}
    Carr, P. and Madan, D.~B. (1999):
    Option valuation using the fast Fourier transform.
    \emph {Journal of Computational Finance} {\bf 2}(4), 61--73.

\bibitem{CGMY}
    Carr, P., Geman, H., Madan, D.~B. and Yor, M. (2002):
    The fine structure of asset returns: An empirical investigation.
    \emph {Journal of Business} {\bf 75}(2), 305--332.

\bibitem{Cont-Tankov}
    Cont, R. and Tankov, P. (2004):
    \emph {Financial modelling with jump processes.}
    Chapman and Hall / CRC Press, London.

\bibitem{FS}
    F\"ollmer, H. and Schweizer, M. (1991):
    Hedging of contingent claims under incomplete information.
    In: M. H. A. Davis and R. J. Elliott (eds.), "Applied Stochastic Analysis", Stochastics Monographs, vol. 5, Gordon and Breach, London/New York, 389--414. 

\bibitem{Gerber}
    Gerber, H.~U. and Shiu, E.~S.~W. (1994):
    Option pricing by Esscher transforms.
    \emph {Trans. Soc. Actuar.} {\bf XLVI}, 98--140.

\bibitem{Grandits}
    Grandits, P. (1999):
    The $p$-optimal martingale measure and its asymptotic relation with the minimal-entropy martingale measure.
    \emph {Bernoulli} {\bf 5}(2), 225-247.

\bibitem{Grandits-R}
    Grandits, P. and Rheinl\"ander, T. (2002):
    On the minimal entropy martingale measure.
    \emph {Annals of Probability} {\bf 30}(3), 1003--1038.

\bibitem{Hubalek}
    Hubalek, F. and Sgarra, C. (2006):
    Esscher transforms and the minimal entropy martingale measure for exponential L\'evy models.
    \emph {Quantitative Finance} {\bf 6}(2), 125--145.

\bibitem{Jacod-Shiryaev}
  Jacod, J. and Shiryaev, A.~N. (2003):
  \textit{Limit theorems for stochastic processes}.
  Springer, Berlin.

\bibitem{Jeanblanc}
    Jeanblanc, M., Kl\"oppel, S. and Miyahara, Y. (2007):
    Minimal $f^q$-martingale measures for exponential L\'evy
    processes.
    \emph {Annals of Applied Probability} {\bf 17}(5-6), 1615--1638.

\bibitem{Teichmann}
    Keller-Ressel, M., Papapantoleon, A. and Teichmann, J. (2013):
    The affine LIBOR models. \emph {Mathematical Finance} {\bf 23}(4), 627--658.

\bibitem{Bianchi-1}
    Kim, Y.~S., Rachev, S.~T., Chung, D.~M. and Bianchi, M.~L. (2009):
    The modified tempered stable distribution, GARCH models and option pricing.
    \emph {Probability and Mathematical Statistics} {\bf 29}(1), 91--117.

\bibitem{Kohlmann-2008}
    Kohlmann, M. and Xiong, D. (2008):
    The minimal entropy and the convergence of the $p$-optimal martingale measures in a general jump model.
    \emph {Stochastic Analysis and Applications} {\bf 26}(5), 941--977.

\bibitem{Koponen}
    Koponen, I. (1995):
    Analytic approach to the problem of convergence of truncated L\'{e}vy flights towards the Gaussian stochastic process.
    \emph {Physical Review E} {\bf 52}, 1197--1199.

\bibitem{Kuechler-Soerensen}
    K{\"u}chler, U. and S\o{}rensen, M. (1997):
    \emph {Exponential families of stochastic processes.}
    Springer, New York.

\bibitem{Kuechler-Tappe}
    K\"uchler, U. and Tappe, S. (2008):
    Bilateral Gamma distributions and processes in financial mathematics.
    \emph {Stochastic Processes and their Applications} {\bf 118}(2), 261--283.

\bibitem{Kuechler-Tappe-shapes}
    K\"uchler, U. and Tappe, S. (2008):
    On the shapes of bilateral Gamma densities.
    \emph {Statistics and Probability Letters} {\bf 78}(15), 2478--2484.

\bibitem{Kuechler-Tappe-pricing}
    K\"uchler, U. and Tappe, S. (2009):
    Option pricing in bilateral Gamma stock models.
    \emph {Statistics and Decisions} {\bf 27}, 281--307.

\bibitem{Kuechler-Tappe-TS}
    K\"uchler, U. and Tappe, S. (2013):
    Tempered stable distributions and processes.
    \emph {Stochastic Processes and their Applications} {\bf 123}(12), 4256--4293.

\bibitem{Lewis}
    Lewis, A.~L. (2001):
    A simple option formula for general jump-diffusion and other
    exponential L\'evy processes.
    \emph {Envision Financial Systems and OptionCity.net}\\ 
    {\tt (http://optioncity.net/pubs/ExpLevy.pdf)}

\bibitem{Madan}
    Madan, D.~B. (2001):
    Purely discontinuous asset pricing processes.
    In: Jouini, E., Cvitani{\v c}, J. and Musiela, M.
    (Eds.), pp. 105--153.
    \emph {Option Pricing, Interest Rates and Risk Management.}
    Cambridge University Press, Cambridge.

\bibitem{Madan-1990}
    Madan, D.~B. and Seneta, B. (1990):
    The VG model for share market returns.
    \emph {Journal of Business} {\bf 63}, 511--524.

\bibitem{Mercuri}
     Mercuri, L. (2008):
     Option pricing in a Garch model with tempered stable innovations.
     \emph {Finance Research Letters} {\bf 5}, 172--182 

\bibitem{Bianchi-book}
  Rachev, S.~T., Kim, Y.~S., Bianchi, M.~L. and Fabozzi, F.~J. (2011):
  \textit{Financial models with L\'{e}vy processes and volatility clustering.} John Wiley \& Sons, Inc., Hoboken, New Jersey.

\bibitem{Rosinski}
    Rosi\'{n}ski, J. (2007):
    Tempering stable processes.
    \emph {Stochastic Processes and their Applications} {\bf 117}(6), 677--707.

\bibitem{Santacroce-2005}
    Santacroce, M. (2005):
    On the convergence of the $p$-optimal martingale measures to the minimal entropy martingale measure.
    \emph {Stochastic Analysis and Applications } {\bf 23}(1), 31--54.

\bibitem{Sato}
    Sato, K. (1999):
    \emph {L\'evy processes and infinitely divisible distributions.}
    Cambridge studies in advanced mathematics, Cambridge.

\bibitem{Sztonyk}
     Sztonyk, P. (2010):
     Estimates of tempered stable densities.
     \emph {Journal of Theoretical Probability} {\bf 23}(1), 127--147.

\bibitem{Zhang}
     Zhang, S. and Zhang, X. (2009):
     On the transition law of tempered stable Ornstein-Uhlenbeck processes.
     \emph {Journal of Applied Probability} {\bf 46}(3), 721--731.

\end{thebibliography}
\end{document}